\documentclass[11pt]{article}
\usepackage[margin=1in]{geometry}

\usepackage{authblk}

\usepackage{amsmath, amsthm, amsfonts, amssymb}

\usepackage{hyphenat}

\usepackage{commath}

\usepackage{mathtools}

\usepackage{thm-restate}

\numberwithin{equation}{section}

\usepackage{enumitem}

\usepackage{mdwlist}

\usepackage{csvsimple}

\usepackage{soul}

\usepackage{environ}

\newcommand{\repeatlemma}[1]{%
	\begingroup
	\renewcommand{\thelemma}{\ref{#1}}%
	\expandafter\expandafter\expandafter\lemma
	\csname replemma@#1\endcsname
	\endlemma
	\endgroup
}

\NewEnviron{replemma}[1]{%
	\global\expandafter\xdef\csname replemma@#1\endcsname{%
		\unexpanded\expandafter{\BODY}%
	}%
	\expandafter\lemma\BODY\unskip\label{#1}\endlemma
}

\usepackage{algorithm}
\usepackage[noend]{algpseudocode}

\makeatletter
\algnewcommand{\StateInd}[1]{\State
	\parbox[t]{\dimexpr\linewidth-\ALG@thistlm}{\strut\hangindent=\algorithmicindent\hangafter=1#1}}

\newlength{\trianglerightwidth}
\algnewcommand{\LineComment}[1]{\Statex \hskip\ALG@thistlm $\triangleright$ \colorcomment{#1}}
\algnewcommand{\LineCommentCont}[1]{\Statex \hskip\ALG@thistlm%
	\parbox[t]{\dimexpr\linewidth-\ALG@thistlm}{\hangindent=\trianglerightwidth \hangafter=1 \strut$\triangleright$ \colorcomment{#1}\strut}}

\algnewcommand{\Input}[1]{\StateInd{\textbf{Input:} #1}}
\algnewcommand{\Output}[1]{\StateInd{\textbf{Output:} #1}}
\makeatother

\newcommand{\colorcomment}[1]{\textcolor{teal}{#1}}

\usepackage{setspace}

\usepackage{makecell}

\usepackage[pdfpagemode={UseOutlines}, bookmarks=true, bookmarksopen=true, bookmarksopenlevel=0, bookmarksnumbered=true, hypertexnames=false, colorlinks, linkcolor=black, citecolor=blue, pdfstartview={FitV}, unicode, breaklinks=true, urlcolor=black]{hyperref}

\usepackage[utf8]{inputenc}
\usepackage[english]{babel}
\usepackage{lmodern}
\usepackage[T1]{fontenc}
\usepackage{graphicx}
\usepackage{color}

\usepackage[noabbrev, capitalize]{cleveref}
\crefname{claim}{Claim}{Claims}
\Crefname{claim}{Claim}{Claims}

\crefname{algorithm}{Algorithm}{Algorithms}
\Crefname{algorithm}{Algorithm}{Algorithms}

\DeclareMathOperator{\poly}{poly}
\DeclareMathOperator{\polylog}{polylog}

\let\oldnl\nl
\newcommand{\nonl}{\renewcommand{\nl}{\let\nl\oldnl}}

\newcommand{\Funcf}{f}

\newcommand{\GroundSet}{\mathcal{N}}
\newcommand{\GroundSubset}{\mathcal{M}}
\newcommand{\GroundSetSize}{n}
\newcommand{\GroundSubsetSize}{\GroundSetSize'}
\newcommand{\OptSet}{\mathcal{O}}
\newcommand{\OptValue}{\textnormal{OPT}}

\newcommand{\Elem}{x}

\newcommand{\OptElem}{o}
\newcommand{\SetA}{S}
\newcommand{\SetB}{T}
\newcommand{\SetS}{S}
\newcommand{\SetT}{T}
\newcommand{\SetQ}{Q}
\newcommand{\Error}{\varepsilon}
\newcommand{\ErrorA}{\Error'}
\newcommand{\ErrorB}{\hat{\Error}}

\newcommand{\ApproxParam}{\beta}
\newcommand{\ApproxParamA}{\alpha}

\newcommand{\TopElem}{\mathcal{A}}
\newcommand{\BestSetTopElem}{A^{*}}

\newcommand{\Card}{k}
\newcommand{\PSep}{p}

\newcommand{\SampleIter}{u}
\newcommand{\SampleIndex}{j}
\newcommand{\SampleElem}{t}

\newcommand{\NumSets}{\ell}
\newcommand{\SetIndex}{l}
\newcommand{\OuterIterIndex}{\IterIndex}
\newcommand{\InnerIterIndex}{j}
\newcommand{\PropExp}{h}
\newcommand{\RemElemsDefer}{\RemElems}
\newcommand{\Curv}{\kappa}
\newcommand{\Homo}{\mu}

\newcommand{\NumRandVar}{n}
\newcommand{\RandVarIndex}{i}
\newcommand{\RandVarProb}{p}
\newcommand{\Mean}{\mu}
\newcommand{\Dev}{d}

\newcommand{\FirstSucc}{A}

\newcommand{\NumRemElemsInter}{Z}
\newcommand{\DepTrial}{X}
\newcommand{\depTrial}{x}
\newcommand{\IndepTrial}{Y}
\newcommand{\ProbTrialSucc}{\eta}
\newcommand{\NumSucc}{c}
\newcommand{\RandVarMaxIndex}{t}

\newcommand{\LS}{\textsc{LinearSeq}}
\newcommand{\LSa}{LS}
\newcommand{\PGB}{\textsc{ParallelGreedyBoost}}
\newcommand{\PGBa}{PGB}

\newcommand{\LSPGBa}{LS+PGB}
\newcommand{\TS}{\textsc{ThresholdSeq}}
\newcommand{\TSa}{TS}
\newcommand{\LALS}{\textsc{LowAdapLinearSeq}}
\newcommand{\LALSa}{LALS}

\newcommand{\LST}{\textsc{LinearSeq}}
\newcommand{\LSTa}{LS}
\newcommand{\PGS}{\textsc{GreedySampling}}
\newcommand{\PGSa}{GS}
\newcommand{\TBS}{\textsc{ThresholdBlockSeq}}
\newcommand{\TBSa}{TBS}
\newcommand{\LSTPGS}{\LST+\PGS}
\newcommand{\LSTPGSa}{\LSTa+\PGSa}
\newcommand{\LALST}{\textsc{LowAdapTop}}
\newcommand{\LALSTa}{LAT}

\newcommand{\SR}{SingleRun}

\newcommand{\AS}{\textsc{Adaptive-Sampling}}

\newcommand{\OptLB}{\Gamma}
\newcommand{\Threshold}{\tau}
\newcommand{\Prob}{\delta}
\newcommand{\GoodElems}{\mathcal{G}}
\newcommand{\MinSizeForSample}{m}
\newcommand{\MinSizeForSampleMax}{\MinSizeForSample_{\text{max}}}
\newcommand{\IterIndex}{i}

\newcommand{\Funcg}{g}
\newcommand{\RemElems}{E}

\newcommand{\NumIters}{\textit{numOut}}
\newcommand{\PrefixIndices}{\Lambda}
\newcommand{\GoodPrefixIndices}{\PrefixIndices^{*}}
\newcommand{\PrefixIndex}{\lambda}
\newcommand{\PrefixElemIndex}{v}
\newcommand{\PrefixIndexGood}{\PrefixIndex^{*}}
\newcommand{\PrefixIndexMax}{\PrefixIndex_{\text{max}}}
\newcommand{\MaxPrefixLength}{\textit{maxSize}}
\newcommand{\PrefixIndexBest}{\PrefixIndex_{\text{best}}}
\newcommand{\PrefixIndexReq}{\PrefixIndex_{\text{req}}}
\newcommand{\Prefix}{P}
\newcommand{\NumBlocks}{\textit{numIn}}
\newcommand{\Block}{B}
\newcommand{\BlockFiltered}{\Block^{*}}

\newcommand{\CardA}{\Card'}
\newcommand{\ReduProp}{\phi}
\newcommand{\ReduPropVal}{0.11879}

\newcommand{\BlockElem}{b}
\newcommand{\GeoExp}{h}

\newcommand{\ProbFail}{\delta}

\newcommand{\ElemGain}{w}

\newcommand{\SampleFailure}{Sample failure}
\newcommand{\PrefixFailure}{Prefix failure}
\newcommand{\InnerFailure}{Inner failure}
\newcommand{\OuterFailure}{Outer failure}
\newcommand{\TBSFailure}{\TBSa{} failure}

\newcommand{\sampleFailure}{sample failure}
\newcommand{\prefixFailure}{prefix failure}
\newcommand{\innerFailure}{inner failure}
\newcommand{\outerFailure}{outer failure}

\newcommand{\PGSFailure}{\PGSa{} failure}

\newcommand{\SMCC}{SMCC}

\newcommand{\Constc}{c}
\newcommand{\Constcp}{\Constc'}

\newtheorem{theorem}{Theorem}

\newtheorem{lemma}{Lemma}

\theoremstyle{definition}
\newtheorem{definition}{Definition}
\theoremstyle{remark}
\newtheorem{claim}{Claim}
\theoremstyle{example}

\begin{document}
	
	\title{Fast Parallel Algorithms for Submodular\\ $p$-Superseparable Maximization}
	
	\author{Philip Cervenjak\footnote{Corresponding author.}}
	\author{Junhao Gan}
	\author{Anthony Wirth}
	
	\affil{School of Computing and Information Systems\\The University of Melbourne, Parkville VIC, Australia}
	
	\affil{\href{mailto:pcervenjak@student.unimelb.edu.au}{\texttt{pcervenjak@student.unimelb.edu.au}}, 
		\href{mailto:junhao.gan@unimelb.edu.au}{\texttt{junhao.gan@unimelb.edu.au}},
		\href{mailto:awirth@unimelb.edu.au}{\texttt{awirth@unimelb.edu.au}}
	}
	
	\date{}
	
	\maketitle              

	\begin{abstract}
		Maximizing a non-negative, monontone, submodular function $\Funcf$ over $\GroundSetSize$ elements under a cardinality constraint $\Card$ (\SMCC) is a well-studied NP-hard problem. It has important applications in, e.g., machine learning and influence maximization. Though the theoretical problem admits polynomial-time approximation algorithms, solving it in practice often involves frequently querying submodular functions that are expensive to compute. This has motivated significant research into designing parallel approximation algorithms in the \emph{adaptive complexity model}; adaptive complexity (adaptivity) measures the number of sequential rounds of $\poly(\GroundSetSize)$ function queries an algorithm requires. The state-of-the-art algorithms can achieve $(1-\frac{1}{e}-\Error)$-approximate solutions with $O(\frac{1}{\Error^2}\log \GroundSetSize)$ adaptivity, which approaches the known adaptivity lower-bounds. However, the $O(\frac{1}{\Error^2} \log \GroundSetSize)$ adaptivity only applies to maximizing worst-case functions that are unlikely to appear in practice. Thus, in this paper, we consider the special class of \emph{$\PSep$-superseparable} submodular functions, which places a reasonable constraint on $\Funcf$, based on the parameter $\PSep$, and is more amenable to maximization, while also having real-world applicability. Our main contribution is the algorithm \LSTPGSa{}, a finer-grained version of the existing \LSPGBa{} algorithm, designed for instances of \SMCC{} when $\Funcf$ is $\PSep$-superseparable; it achieves an expected $(1-\frac{1}{e}-\Error)$-approximate solution with $O(\frac{1}{\Error^2}\log(\PSep \Card))$ adaptivity \emph{independent of $\GroundSetSize$}. Additionally, unrelated to $\PSep$-superseparability, our \LSTPGSa{} algorithm uses only $O(\frac{\GroundSetSize}{\Error} + \frac{\log \GroundSetSize}{\Error^2})$ oracle queries, which has an improved dependence on $\Error^{-1}$ over the state-of-the-art \LSPGBa{}; this is achieved through the design of a novel thresholding subroutine.
		
		\vspace*{\baselineskip}
	\end{abstract}
	\newpage
	\section{Introduction}
	Submodular functions are an important class of set functions that capture a wide range of real-world applications that, informally,  exhibit the property of ``diminishing marginal gains'' or ``diminishing returns''. 
	In this paper, we consider maximizing {\em non-negative}, {\em monotone}, {\em submodular} functions $\Funcf : 2^{\GroundSet} \rightarrow \mathbb{R}_{\geq 0}$, defined on a ground set $\GroundSet$ of $\GroundSetSize$ elements, under a cardinality constraint (\SMCC).
	The goal of \SMCC~is to select a subset $S \subseteq \GroundSet$ of size $|\SetS| \leq \Card$ that maximizes $\Funcf(\SetS)$. 
	As a convention in the literature, we assume that, for any $\SetS \subseteq \GroundSet$, the value of $\Funcf(\SetS)$ can {\em only} be accessed via {\em queries} to a {\em value oracle}.
	
	Solving \SMCC{} is important for a wide range of applications, including machine learning (e.g., active learning \cite{Wei2015}, clustering \cite{Dueck2007}, data summarization \cite{mirzasoleiman2016fast}, and feature selection \cite{Khanna2017}), information gathering \cite{krause2011submodularity}, network monitoring \cite{Leskovec2007}, sensor placement \cite{krause2008near}, and influence maximization~\cite{kempe2003maximizing}.
	
	\paragraph{The Greedy Algorithm.}
	As is true for most interesting variants of submodular maximization, \SMCC{} is unfortunately an NP-hard problem. 
	Even worse, the best approximation that can be achieved with a polynomial number of oracle queries is $1-{1}/{e}$, unless $\text{P} = \text{NP}$~\cite{Nemhauser1978}.
	Interestingly, the ``best'' such approximation ratio can be achieved by a simple {\em greedy} algorithm~\cite{nemhauser1978analysis}. Specifically, the greedy algorithm starts with a solution $\SetS = \varnothing$ and performs $\Card$ iterations, in each of which the element with the largest {\em marginal gain} with respect to $\SetS$ is added to $\SetS$.
	In its raw form, the greedy algorithm queries $\Funcf$ $O(\Card \GroundSetSize)$ times, and it is strongly {\em sequential}: it has to perform $\Card$ iterations one by one.
	
	\paragraph{The Adaptive Complexity Model.}
	In practice, querying the oracle for a set's value, i.e., evaluating $\Funcf(\SetS)$, can be time consuming and it is often the main bottleneck of the overall running time of an algorithm.
	This has motivated significant research into designing {\em parallelizable} algorithms for \SMCC{} under the {\em adaptive complexity model}~\cite{balkanski2018adaptive,Balkanski2018,ene2019submodular,chekuri2019submodular,fahrbach2019submodular,balkanski2019exponential,kazemi2019submodular,breuer2020fast,chen2021best}, where the {\em efficiency} of an algorithm is measured by {\em the number of queries} and in each round, an algorithm is allowed to perform a polynomial number, $\poly(\GroundSetSize)$, of independent oracle queries in parallel.
	Each such round is called an {\em adaptive round} and the total number of adaptive rounds required is called the {\em adaptive complexity} (or {\em adaptivity}) of the algorithm.
	The smaller an algorithm's adaptivity is, the more parallelizable the algorithm is.
	Clearly, the adaptivity of the greedy algorithm is $O(\Card)$.

	\paragraph{The State-of-the-Art Adaptive Algorithm.}
	The goal of all existing adaptive algorithms~\cite{balkanski2018adaptive,Balkanski2018,ene2019submodular,chekuri2019submodular,fahrbach2019submodular,balkanski2019exponential,kazemi2019submodular,breuer2020fast,chen2021best} for \SMCC{} is to beat the $O(k)$ adaptivity bound of the standard greedy algorithm, and ideally, to beat the query complexity $O(kn)$ at the same time. 
	The state-of-the-art algorithm, due to Chen et al.~\cite{chen2021best}, called \LSPGBa{}, achieves a $(1-\frac{1}{e}-\Error)$-approximation with an adaptive complexity of $ O(\frac{1}{\Error^2}\log(\frac{\GroundSetSize}{\Error}))$ and a query complexity 
	of $O(\frac{\GroundSetSize}{\Error^2})$.
	Assuming $k \in \omega(\frac{1}{\Error^2} \log \GroundSetSize)$,
	\LSPGBa{} achieves $o(k)$ adaptivity and $o(kn)$ query complexity simultaneously, improving the naive greedy algorithm.
	
	\paragraph{Known Lower Bounds.}
	For \SMCC{}, Balkanski and Singer \cite{balkanski2018adaptive} initially proved that $\Omega(\frac{\log \GroundSetSize}{\log \log \GroundSetSize})$ adaptive rounds are required to achieve a \sloppy $\frac{1}{\log \GroundSetSize}$\nobreakdash-approximation. 
	Li et al. \cite{li2020polynomial} later proved lower bounds for achieving a $(1-\frac{1}{e}-\Error)$-approximation in two cases of $\Error > 0$: when $\Error > \frac{\Constcp}{\log \GroundSetSize}$, $\Omega(\frac{1}{\Error})$ rounds are required; when $\Error < \frac{\Constcp}{\log \GroundSetSize}$, $\Omega( \frac{\log^{2/3} \GroundSetSize}{\Error^{1/3}} )$ rounds are required, where $\Constcp$ is an absolute constant. Kuhnle \cite{kuhnle2021quick} proved that $\Omega(\frac{\GroundSetSize}{\Card})$ queries are required to achieve a constant-factor approximation, even when queries can be made to infeasible sets.
	
	\paragraph{Our Research Question.}
	In spite of the aforementioned progress, \SMCC{} is a general problem formulation and, thus, captures difficult problem instances that are not likely to appear in practice.
	Analogously, although the greedy algorithm for Set Cover achieves only an~$O(\log n)$ approximation factor on~$n$ elements, the well known tight example is bespoke, and in practice greedy performs well~\cite{Grossman1997}.
	It would be of theoretical and practical interest if there were a useful class of submodular functions that can be maximised in fewer adaptive rounds than what is needed for the worst-case functions, especially since real-world submodular functions can be computationally expensive to query. This motivates our main research question:
	\begin{center}
		\emph{Is there an interesting class of \SMCC{} instances  that admits\\ algorithms with $o(\log \GroundSetSize)$ adaptive complexity,\\ while achieving reasonable approximation?}
	\end{center}
	
	We address our research question by considering the class of \emph{$\PSep$-superseparable} submodular functions ($\PSep \in [1, \GroundSetSize]$ is a class parameter); in particular, we design highly parallel approximation algorithms for \SMCC{} when $\Funcf$ is assumed to be $\PSep$-superseparable. This class of functions belongs in the super-class of $\PSep$-\emph{separable} submodular functions\footnote{Our work does not focus on the other two classes of $\PSep$-separable functions, which are $\PSep$-subseparable and rev-$\PSep$-subseparable functions.}, introduced by Skowron \cite{SKOWRON201765} for the purpose of showing fixed-parameter tractable (FPT) algorithms for \SMCC{}. As Skowron shows, $\PSep$-superseparable submodular functions capture several useful real-world functions, such as those found in election and recommendation systems. They also capture the Max-$\Card$ (Weighted) Coverage problem with element frequencies upper-bounded by $\PSep$.
	We now outline our contributions one by one.
	\subsection{Our Contributions}
	\paragraph{Parallel Algorithm for $\PSep$-Superseparable \SMCC{}.} \label{sec:LSTPGS contributions}
	Our first contribution is the algorithm \LSTPGS{} (\LSTPGSa) and its subroutine \PGS{} (\PGSa) for $\PSep$-superseparable \SMCC{}, with theoretical guarantees stated in \cref{theorem:LSTPGS results,theorem:PGS results} respectively.
	
	\LSTPGSa{} is essentially a finer-grained version, parameterised by $\PSep$, of \LSPGBa{} \cite{chen2021best} that exploits the $\PSep$-superseparability of $\Funcf$ to achieve an adaptive complexity of $O( \frac{1}{\Error^2} \log(\frac{\PSep \Card}{\Error}))$. 
	For example, assuming $\PSep, \Card \in O(\polylog \GroundSetSize)$ as well as $\Card \in \omega(\frac{1}{\Error^2} \log \GroundSetSize)$ (since otherwise the greedy algorithm would have adaptivity as good as the existing algorithms), the adaptivity of our algorithm is bounded by $O( \frac{1}{\Error^2}\log \log \GroundSetSize )$. Under this setting, our algorithm's adaptive complexity beats the $O( \frac{1}{\Error^2} \log \GroundSetSize )$ adaptive complexity of the existing algorithms for general \SMCC{}, as well as the $O(\Card) = O(\polylog \GroundSetSize)$ adaptive complexity of the greedy algorithm. We summarize the performance guarantees of \LSTPGSa{}, compared with \LSPGBa{}, in the \cref{table:kpi-smcc} below. For simplicity and fair comparisons, we assume $\Card \in \omega(\frac{1}{\Error^2} \log \GroundSetSize)$ in \cref{table:kpi-smcc}.
	
	\begin{table}
		\caption{Key performance indicators for \SMCC{} algorithms.}
		\label{table:kpi-smcc}
		\begin{center}
			\vspace*{-2mm}
			\begin{tabular}{ |c||c|c|c| }
				\hline
				\textbf{Algorithm} & \textbf{Approx.} & \textbf{Adaptivity} & \textbf{Expected queries} \\
				\hline\hline
				\makecell{\LSPGBa{} \cite{chen2021best}} & $1-\frac{1}{e}-\Error$ & $O\left( \frac{1}{\Error^2} \log{\GroundSetSize} \right)$ & $O\left( \frac{\GroundSetSize}{\Error^2} \right)$ \\
				\hline
				\makecell{general \LSTPGSa{} \\ (\cref{theorem:LSTPGS general results}, ours)} & $ 1-\frac{1}{e}-\Error $ & $O\left( \frac{1}{\Error^2} \log{\GroundSetSize} \right)$ & $O\left( \frac{\GroundSetSize}{\Error} \right)$ \\
				\hline
				\makecell{$\PSep$-supersep. \LSTPGSa{} \\ (\cref{theorem:LSTPGS results}, ours)} & \makecell{expected\\$1-\frac{1}{e}-\Error$} & $O\left( \frac{1}{\Error^2} \log\left({\PSep \Card}\right) \right)$ & $O\left( \GroundSetSize + \frac{\PSep\Card}{\Error^2}\right)$ \\ 
				\hline
			\end{tabular}
		\end{center}
	\end{table}
	
	Additionally, for general \SMCC{}, \LSTPGSa{} has an expected query complexity of $O( \frac{\GroundSetSize}{\Error} + \frac{\log \GroundSetSize}{\Error^2} )$. This improves the dependence on $\Error^{-1}$ in the $O( \frac{\GroundSetSize}{\Error^2} )$ expected query complexity of \LSPGBa{} \cite{chen2021best}. The improved query complexity is due to \PGSa{} using a novel thresholding procedure, which we outline next. We formally state the theoretical guarantees of \LSTPGSa{} for general \SMCC{} in \cref{theorem:LSTPGS general results}.
	
	\paragraph{Parallel Thresholding Procedure for \SMCC{}.} \label{sec:TBS contributions}
	Our second contribution is the procedure \TBS{} (\TBSa), with theoretical guarantees stated in \cref{theorem:TBS results}. \TBSa{} is used by \PGSa{} for the task of selecting a set of elements whose average marginal gain approximately satisfies a given threshold $\Threshold$.
	
	The significance of \TBSa{} is that, given $\GroundSubsetSize$ input elements and an error term $\ErrorA$, it achieves an expected query complexity of $O(\GroundSubsetSize + \frac{\log \GroundSubsetSize}{\ErrorA})$.
	This improves the dependence on $\ErrorA^{-1}$ in the query complexity: existing procedures for the same task \cite{fahrbach2019submodular,10.1145/3313276.3316304,kazemi2019submodular,chen2021best} perform $O( \frac{\GroundSubsetSize}{\ErrorA})$ queries. Note that \TBSa{} does not rely on $\PSep$-superseparability for its improved query complexity; indeed, the improved query complexity of \TBSa{} is what leads to the $O( \frac{\GroundSetSize}{\Error} + \frac{\log \GroundSetSize}{\Error^2} )$ expected query complexity of \LSTPGSa{} for general \SMCC{}. This also means that \TBSa{} can replace the existing procedures that are used as subroutines for solving general \SMCC{}. We summarize the guarantees of \TBSa{} in \cref{table:kpi-thresh} below ($\MinSizeForSample$ is an input parameter satisfying $\GroundSubsetSize \leq \MinSizeForSample \leq \GroundSetSize$, and $\ProbFail$ is a failure probability term).
	
	\paragraph{Simple Parallel Algorithm for $\PSep$-Superseparable \SMCC{}.} \label{sec:simple contributions}
	Finally, we introduce \LALST{} (\LALSTa) for $\PSep$-superseparable \SMCC{}, with theoretical guarantees stated in \cref{theorem:simple results}. \LALSTa{} works by simply running the existing procedure \LALS{} \cite{chen2021best} on the set of top-$ \lceil \frac{\PSep\Card}{1-\ApproxParamA} + \Card \rceil $ elements by value. \LALSTa{} achieves an $\ApproxParamA (5 + O(\Error))^{-1}$-approximation, an adaptive complexity of $O( \frac{1}{\Error^3} \log (\frac{\PSep}{1-\ApproxParamA}) )$, and a query complexity of $ O( \frac{1}{1-\ApproxParamA} ( \frac{\PSep \Card}{\Error^3} + \frac{\PSep}{\Error^4} ) ) $.
	
	\begin{table}
		\caption{Key performance indicators for parallel thresholding procedures.}
		\label{table:kpi-thresh}
		\begin{center}
			\vspace*{-2mm}
			\begin{tabular}{ |c||c|c| }
				\hline
				\textbf{Procedure} & \textbf{Adaptivity} & \textbf{Expected queries} \\
				\hline\hline
				\makecell{\TS{} \\ (Theorem 2 of \cite{chen2021best})} & $O\left( \frac{1}{\ErrorA} \log\left(\frac{\MinSizeForSample}{\ProbFail}\right) \right)$ & $O\left( \frac{\GroundSubsetSize}{\ErrorA} \right)$ \\
				\hline
				\makecell{\TBS{} \\ (\cref{theorem:TBS results}, ours)} & $O\left( \frac{1}{\ErrorA} \log\left(\frac{\MinSizeForSample}{\ProbFail}\right) \right)$ & $O\left( \GroundSubsetSize + \frac{\log \GroundSubsetSize}{\ErrorA} \right)$ \\
				\hline
			\end{tabular}
		\end{center}
		\vspace*{-\baselineskip}
	\end{table}
	\subsection{Our Techniques}
	Our algorithms are based on the state-of-the art framework of Chen et al. \cite{chen2021best}, particularly the \LSPGBa{} algorithm.
	Importantly, our two key techniques exploit $\PSep$-superseparability to achieve adaptive complexities \emph{independent of $\GroundSetSize$}, rather depending on parameters $\PSep$ and $\Card$.
	
	The first key technique is to run an existing algorithm on a limited number of ``top-valued'' elements so as to bound the adaptive complexity of the algorithm. It follows from the $\PSep$-superseparability of $\Funcf$ that the set of top-$ \lceil \frac{\PSep\Card}{1-\ApproxParamA} + \Card \rceil $ valued elements contains a $\Card$-size $\ApproxParamA$\nobreakdash-approximation of the optimal solution, leading to a good approximation overall. This result is stated in \cref{lemma:top elems opt value} and follows from Theorem 1 of Skowron \cite{SKOWRON201765}. We use this technique in our main algorithm \LSTPGSa{} and our algorithm \LALSTa{}.
	
	The second key technique is to uniformly-at-random sample elements from a sufficiently large set $\GoodElems$ of valuable elements. The $\PSep$-superseparability of $\Funcf$ ensures that newly sampled elements, on average, remain valuable even after previously sampled elements are added to a solution, bypassing the need to make sequential oracle queries; this is formally stated in \cref{lemma:expected gain of sample}. We use this technique in \PGSa{}, a subroutine of \LSTPGSa{}.
	
	Furthermore, independent of $\PSep$-superseparability, we develop a new ``element filtering'' technique, used by \TBSa{} to achieve its improved expected query complexity of $O(\GroundSubsetSize + \frac{\log \GroundSubsetSize}{\ErrorA})$. The key insight in \TBSa{} is to avoid repeated filtering queries on \emph{all} remaining elements, and instead mainly perform the filtering queries on small random samples or ``blocks'' of the remaining elements.
	
	\subsection{Paper Structure} We give additional related work in \cref{sec:related work}, and then present preliminaries in \cref{sec:preliminaries}, including the definition and intuition of $\PSep$-superseparability, and an overview of the state-of-the-art algorithm \LSPGBa{}. We then give the descriptions and analyses of \LSTPGSa{} and \PGSa{} in \cref{sec:parallel alg p-super}, \TBSa{} in \cref{sec:thresholdblockseq}, and \LALSTa{} in \cref{sec:simple}. Finally, we give some conclusions in \cref{sec:conclusions}.
	
	\subsection{Additional Related Work} \label{sec:related work}
	\paragraph{$\PSep$-Separable Submodular Functions.} \label{sec:p-sep functions}
	Skowron originally proposed $\PSep$-separable functions \cite{SKOWRON201765} for the purpose of showing that \SMCC{} admits fixed-parameter tractable (FPT) algorithms with respect to parameters $\Card$ and $\PSep$, achieving arbitrarily good approximation factors. Moreover, Skowron showed that \SMCC{} for $\PSep$-superseparable or $\PSep$-subseparable $\Funcf$ captures a range of real-world maximization problems. For example, it captures a general problem framework where one must select $\Card$ items to maximize the total satisfaction of agents that can only approve of $\PSep$ items; this framework models problems in multiwinner election systems, recommendation systems, and facility location. These two variants of \SMCC{} also capture the Max-$\Card$ (Weighted) Coverage problem where each element's frequency is at most $\PSep$, i.e., each element appears in at most $\PSep$ sets. On the other hand, \SMCC{} for rev-$\PSep$-subseparable $\Funcf$ captures the same problem but each element's frequency is at least $\PSep$. 
	
	Skowron proposed an FPT algorithm for $\PSep$-superseparable \SMCC{} that returns an $\ApproxParamA$-approximate solution, where $0 < \ApproxParamA < 1$. This algorithm works by constructing the set $\TopElem$ of top-$ \lceil \frac{\PSep\Card}{1-\ApproxParamA} + \Card \rceil $ elements by value and then, by brute-force search, returning the $\Card$-size subset of $\TopElem$ that maximizes $\Funcf$. The $\ApproxParamA$-approximation of the returned subset follows from the $\PSep$-superseparability of $\Funcf$ amongst its other properties, as shown in Theorem 1 of the paper.
	
	Skowron also proposed an FPT algorithm for $\PSep$-subseparable \SMCC{} that returns a solution $\SetS$ satisfying $ \Funcf(\GroundSet) - \Funcf(\SetS) \leq \ApproxParamA( \Funcf(\GroundSet) - \OptValue )$ with probability arbitrarily close to 1, where $\ApproxParamA > 1$. This algorithm works by first defining a procedure `\SR{}'. This procedure builds a solution by performing $\Card$ sequential samples of the ground set, where an element is sampled with probability proportional to its marginal gain to the current solution. By the $\PSep$-subseparability of $\Funcf$ amongst its other properties, \SR{} has a non-negligible probability of sampling $\Card$ optimal solution elements while the partial solutions do not satisfy the required approximation. The main FPT algorithm simply repeats \SR{} sufficiently many times to boost the probability of sampling a solution that is optimal or that satisfies the required approximation.
	
	Further, Skowron showed that, for maximizing a non-negative, monotone, rev-$\PSep$-subseparable function (not necessarily submodular) under a cardinality constraint, the standard greedy algorithm returns a $(1-e^{-\frac{\PSep\Card}{\GroundSetSize}})$-approximation. If the function is additionally assumed to be submodular, then the greedy algorithm returns a $(1-e^{-\max \{ \frac{\PSep\Card}{\GroundSetSize},1 \}})$-approximation.
	
	\paragraph{Curvature and Homogeneity of a Submodular Function.} \label{sec:curv and hom functions}
	\emph{Curvature} is a different property of submodular functions $\Funcf$ that was introduced by Conforti and Cornuejols \cite{Conforti1984} and has since been well-studied \cite{Vondrak2010,Iyer2013,Iyer2013a,Balkanski2016,Sviridenko2017,Balkanski2018}. $\Funcf$ has curvature $\Curv$, where $0 \leq \Curv \leq 1$, iff for each $\SetS \subseteq \GroundSet$ and $\Elem \in \GroundSet \setminus \SetS \colon \Funcf(\Elem \mid \SetS) \geq (1-\Curv)\Funcf(\Elem)$. The main interest in curvature is that, when $\Funcf$ has bounded curvature, approximation factors better than $1-\frac{1}{e}$ can be obtained for a number of submodular optimization problems. For example, Conforti and Cornuejols originally showed that the greedy algorithm achieves a $\frac{1-e^{-\Curv}}{\Curv}$-approximation for \SMCC{} \cite{Conforti1984}. This result was later improved by Sviridenko et al., who provided a $(1-\frac{\Curv}{e})$-approximation algorithm \cite{Sviridenko2017}; moreover, they showed that this is the best approximation possible using a polynomial number of oracle queries.
	
	Interestingly, Balkanski and Singer \cite{Balkanski2018} studied \SMCC{} for $\Funcf$ with bounded curvature in the adaptive complexity model. They gave a modification of their original \AS{} algorithm \cite{balkanski2018adaptive}, achieving an approximation factor arbitrarily close to $\max(1-\Curv, \frac{1}{2})$ in $O( \frac{\log \GroundSetSize}{1-\Curv} )$ adaptive rounds. This suggests a trade-off between the curvature of $\Funcf$ and the adaptive complexity needed to maintain a near $\frac{1}{2}$-approximation. Moreover, they proved that $\Omega(\frac{\log \GroundSetSize}{\log \log \GroundSetSize})$ adaptive rounds are required to achieve a $(1-\Curv + o(1))$-approximation.
	
	Balkanski and Singer additionally studied a property of submodular functions known as \emph{homogeneity}, although this is not as well-studied as curvature. $\Funcf$ is $\Homo$-homogeneous, where $\Homo \geq 0$, iff for all $\Elem \in \GroundSet \colon \Funcf(\Elem) \leq (1+\Homo)\frac{\OptValue}{\Card}$. Balkanski and Singer showed that, for \SMCC{} for $\Homo$-homogeneous $\Funcf$ (with curvature $\Curv$), the approximation factor of \AS{} is strengthened to be arbitrarily close to $1-\frac{\Homo}{2\Homo+1}$ while retaining its adaptive complexity of $O( \frac{\log \GroundSetSize}{1-\Curv} )$.
	
	Although the aforementioned work of Balkanski and Singer is similar in spirit to ours, curvature is not comparable to $\PSep$-separability as pointed out by Skowron. There are $\PSep$-separable functions with good values of $\PSep$ that have bad curvature (large $\Curv$) \cite{SKOWRON201765}, and conversely there are functions with good curvature (small $\Curv$) that have bad values of $\PSep$.
	
	\section{Preliminaries} \label{sec:preliminaries}
	Denote by $ \GroundSet $ the \emph{ground set} of elements. For every set $\SetA \subseteq \GroundSet$ and a \emph{set function} $ \Funcf $, $\Funcf(\SetA)$ is called the \emph{value} of $\SetA$. Given two sets $ \SetA,\SetB \subseteq\GroundSet$, we define $ \Funcf( \SetA \mid \SetB ) $ to be the \emph{marginal gain} of $ \SetA $ to $ \SetB $, i.e., $ \Funcf( \SetA \mid \SetB ) \coloneqq \Funcf(\SetA \cup \SetB) - \Funcf(\SetB) $. When expressing the value, or marginal gain, of a singleton set,~$\{ \Elem \}$, we abuse notation and use~$\Funcf(\Elem)$ and $\Funcf(\Elem \mid \SetS)$, rather than $\Funcf(\{\Elem\})$ and $\Funcf(\{\Elem\} \mid \SetS)$. We denote by $\OptSet$ the optimal solution to the \SMCC{} problem, and $\OptValue \coloneqq \Funcf(\OptSet)$.
	\begin{definition}[Submodular; monotone; non-negative]
		Set function~$ \Funcf $ is \emph{submodular} iff \mbox{$ \forall \SetA, \SetB \subseteq \GroundSet $} such that \mbox{$ \SetA\subseteq \SetB $}, and $ \forall \Elem \in \GroundSet \setminus \SetB \colon
		\Funcf( \Elem \mid \SetA) \geq \Funcf( \Elem \mid \SetB )$; \emph{monotone} iff $ \forall \SetA, \SetB \subseteq \GroundSet \colon \Funcf( \SetA \mid \SetB ) \geq 0$; \emph{non-negative} iff $ \forall \SetA \subseteq \GroundSet \colon \Funcf(\SetA) \geq 0$.
	\end{definition}
	
	\subsection{$\PSep$-Superseparable Functions} \label{sec:p-superseparable}
	Our interest is in the class of $ \PSep $-superseparable functions, \cref{def:p-separable}, introduced by Skowron \cite{SKOWRON201765}.
	A larger $\PSep$ represents a more general class, so a smaller~$\PSep$ yields stronger results.
	When $\Funcf$ is non-negative and submodular, the smallest sensible value for $\PSep$ is $1$; on the other hand, every monotone $\Funcf$ is $\GroundSetSize$-superseparable. We give some background and applications for $\PSep$-separable functions in \cref{sec:p-sep functions}.
	\begin{definition}[$ \PSep $-superseparable set function \cite{SKOWRON201765}]\label{def:p-separable}
		A set function, $ \Funcf $, is $ \PSep $\emph{-superseparable} iff $ \forall \SetA \subseteq \GroundSet\colon $
		\begin{align}
			\sum_{\Elem\in\GroundSet} \Funcf(\Elem \mid \SetA) \geq \sum_{x\in\GroundSet} \Funcf(\Elem) - \PSep \Funcf(\SetA). \label{eqn:p-superseparable}
		\end{align}
	\end{definition}
	
	\paragraph{Intuition of $\PSep$-Superseparable Functions.}
	Intuitively, dividing both sides of \cref{eqn:p-superseparable} by $\GroundSetSize$, the left-hand side becomes the average marginal gain of an element $\Elem$ to $\SetS$, while the right-hand side becomes the average individual value of an element minus the average ``loss'' or ``overlap'' due to $\SetS$. Hence, $\PSep$-superseparability ensures that, {\em on average}, a single element $\Elem$ loses at most $\frac{\PSep}{\GroundSetSize}\Funcf(\SetS)$ from its individual value to give $\Funcf(\Elem \mid \SetS)$.
	
	\subsection{An Overview of the State-of-the-Art Algorithm} \label{sec:overview of LSPGB}
	We now give an overview of \LSPGBa{} by Chen et al.~\cite{chen2021best}, which is the starting point for our algorithm \LSTPGSa{} as proposed in \cref{sec:parallel alg p-super} (see \cref{table:kpi-smcc}). \LSPGBa{} comprises two procedures performed in sequence, where \PGBa{} invokes \TS{} as a subroutine. Our outline includes our notation rather than that of Chen et al. \cite{chen2021best}.
	
	\begin{description}
		\item[\LS{} (\LSa{})] This is a pre-processing procedure whose purpose, given an approximation error $\ErrorB$, is to obtain a value $\OptLB$ satisfying $\OptLB \leq \OptValue \leq \frac{\OptLB}{\ApproxParam}$ for $\ApproxParam = (4 + O(\ErrorB))^{-1}$. It uses $O(\frac{1}{\ErrorB^3}\log \GroundSetSize)$ adaptive rounds and $O( (\frac{1}{\ErrorB \Card} + 1 ) \frac{\GroundSetSize}{\ErrorB^3} )$ expected queries. The quantity~$\ErrorB$ can be set constant without affecting the main approximation error,~$\Error$.
		\item[\PGB{} (\PGBa{})] With the previously obtained $\OptLB$ and $\ApproxParam$, \PGBa{} initializes a threshold $\Threshold$ to upper-bound the average value of an optimal solution element, i.e., $ \Threshold = \frac{\OptLB}{\ApproxParam\Card} \geq \frac{\OptValue}{\Card} $. It then uses a diminishing-threshold strategy to achieve the final $(1-\frac{1}{e}-\Error)$-approximation, while using $O(\frac{1}{\Error^2}\log(\frac{\GroundSetSize}{\Error}))$ adaptive rounds and $O( \frac{\GroundSetSize}{\Error^2} )$ expected queries. It is crucial that $\ApproxParam$ is a constant, as this helps to bound the number of threshold diminutions over a while-loop.
		\item[\TS{} (\TSa{})] For each threshold $\Threshold$ considered and a current solution $\SetS$, \PGBa{} calls \TSa{} to select a set of elements $\SetT$ whose marginal gain to $\SetS$ approximately satisfies $\Threshold |\SetT|$; \PGBa{} then appends $\SetT$ to $\SetS$. Given a failure probability term $\ProbFail$, \TSa{} uses $O( \frac{1}{\ErrorA} \log( \frac{\GroundSetSize}{\ProbFail} ) )$ adaptive rounds and $O( \frac{\GroundSetSize}{\ErrorA} )$ expected queries. \TSa{} performs a loop, in which each iteration appends elements to $\SetT$ using an improved version of the \emph{adaptive sequencing} technique. Adaptive sequencing was introduced by Balkanski et al. for monotone submodular maximization under a matroid constraint \cite{10.1145/3313276.3316304}, and refined in the FAST algorithm by Breuer et al. for \SMCC{} \cite{breuer2020fast}. We describe \TSa{} in more detail in \cref{sec:thresholdblockseq} so as to directly compare it with our improved procedure \TBSa{}.
	\end{description}
	
	\section{Parallel Algorithm for $\PSep$-Superseparable \SMCC{}} \label{sec:parallel alg p-super}
	In this section, we introduce our main parallel approximation algorithm and its components, and formalize the performance guarantees.
	
	\subsection{\LSTPGS{}} \label{sec:overview of LSTPGS}
	We begin with \LSTPGSa{}, pseudocode in \cref{alg:LSTPGS}, for $\PSep$-superseparable \SMCC{}, with theoretical guarantees in \cref{theorem:LSTPGS results} below.
	The performance guarantees of \LSTPGSa{} for the {\em general} \SMCC{} problem are given in \cref{theorem:LSTPGS general results}.
	
	\begin{theorem}\label{theorem:LSTPGS results}
		Let $(\Funcf, \Card)$ be an instance of \SMCC{} where $\Funcf$ is $\PSep$-superseparable. Suppose \LSTPGSa{} (\cref{alg:LSTPGS}) is run such that $0 < \Error < 1-\frac{1}{e}$, $0 < \ErrorB < \frac{1}{2}$, and $0 < \ApproxParamA < 1$, where $\ErrorB$ and $\ApproxParamA$ are constants. Then, with probability $1 - O ( \frac{\Error}{\PSep \Card} )$, \LSTPGSa{} achieves:
		\begin{itemize}
			\item a solution $\SetS$ satisfying $|\SetS| \leq \Card$ and $\mathbb{E}[ \Funcf(\SetS) ] \geq ( 1-\frac{1}{e}-\Error ) \OptValue$,
			\item an adaptive complexity of $O( \frac{1}{\Error^2} \log(\frac{\PSep \Card}{\Error}) )$, and
			\item an expected query complexity of $O( \GroundSetSize + \frac{\PSep\Card}{\Error^2} + \frac{1}{\Error^2}\log( \frac{\PSep \Card}{\Error} ) )$.
		\end{itemize}
	\end{theorem}
	
	\paragraph{Description of \LSTPGSa{}.}
	Based on \LSPGBa{} \cite{chen2021best}, our two key modifications achieve an adaptive complexity dependent on $\PSep$ and $\Card$ rather than~$\GroundSetSize$.
	\begin{enumerate}
		\item To obtain an initial value, $\OptLB$, satisfying $\OptLB \leq \OptValue \leq \frac{\OptLB}{\ApproxParam}$ for a constant $\ApproxParam$, instead of running \LS{} on $\GroundSet$, \LS{} is run only on the set $\TopElem$ of top-$\lceil \frac{\PSep \Card}{1-\ApproxParamA} + \Card \rceil$ elements by individual value, thus using only $O(\frac{1}{\ErrorB^3} \log( \frac{\PSep \Card}{1-\ApproxParamA}))$ adaptive rounds; $\ErrorB$ and $\ApproxParamA$ can be constant.
		\item Instead of running \PGBa{}, our procedure \PGSa{} is run on $\GroundSet$ (taking $\OptLB$ and $\ApproxParam$ in its input). \PGSa{} uses only $O( \frac{1}{\Error^2} \log(\frac{\PSep \Card}{\Error}) )$ adaptive rounds at the expense of returning an \emph{expected} $(1-\frac{1}{e}-\Error)$-approximation.
	\end{enumerate}
	
	\begin{algorithm}
		\caption{} \label{alg:LSTPGS}
		\begin{algorithmic}[1]
			\Procedure{\LSTPGS}{$\Funcf, \GroundSet, \PSep, \Card, \ApproxParamA, \ErrorB, \Error$}
			\Input{value oracle $ \Funcf \colon 2^\GroundSet \rightarrow \mathbb{R}_{\geq 0} $, ground set $ \GroundSet $, parameter $ \PSep $ such that $ \Funcf $ is $ \PSep $-superseparable, cardinality constraint $ \Card $, initial approximation error $ \ErrorB $, initial approximation term $ \ApproxParamA $,  approximation error $\Error$}
			\Output{set $ \SetS $ satisfying $ \mathbb{E}[\Funcf(\SetS)] \geq \left(1-\frac{1}{e}-\Error \right) \OptValue $}
			\State $ \TopElem \gets $ set of top-$ \left\lceil \frac{\PSep\Card}{1-\ApproxParamA} + \Card \right\rceil $ elements $ \Elem \in \GroundSet $ by value $ \Funcf(\Elem) $ \label{line:LSTPGS assign TopElem}
			\State $ \SetQ \gets \textsc{LinearSeq}(\Funcf, \TopElem, \Card, \ErrorB) $ \label{line:LSTPGS run LinearSeq}
			\State $ \OptLB \gets \Funcf(\SetQ) $ \label{line:LSTPGS assign Opt LB}
			\If{$ |\TopElem| < |\GroundSet| $}
			\State $ \ApproxParam \gets \ApproxParamA \left(4 + \frac{4(2-\ErrorB)\ErrorB}{(1-\ErrorB)(1-2\ErrorB)} \right)^{-1} $ \label{line:LSTPGS assign ApproxParam}
			\Else
			\State $ \ApproxParam \gets \left(4 + \frac{4(2-\ErrorB)\ErrorB}{(1-\ErrorB)(1-2\ErrorB)} \right)^{-1} $ \label{line:LSTPGS assign def ApproxParam}
			\EndIf
			\State $ \SetS \gets \textsc{\PGS}(\Funcf, \GroundSet, \PSep, \Card, \ApproxParam, \OptLB, \Error) $ \label{line:LSTPGS run PGS}
			\State \Return{$ \SetS $} \label{line:LSTPGS return solution}
			\EndProcedure
		\end{algorithmic}
	\end{algorithm}
	
	\paragraph{Deriving \cref{theorem:LSTPGS results}.}
	At a high-level, we obtain the guarantees in \cref{theorem:LSTPGS results} by combining the guarantees of running \LS{} on the set $\TopElem$ of top-valued elements and the guarantees of \PGS{} when $\ApproxParam$ is constant (see \cref{theorem:PGS results} below), as they are run one at a time.
	
	We can ensure $ \ApproxParam $ is constant by setting $\ErrorB$ and $\ApproxParamA$ constant and by \cref{lemma:top elems opt value} below, which follows from Theorem 1 of Skowron~\cite{SKOWRON201765}.
	This lemma guarantees that, for $\PSep$-superseparable \SMCC{}, running an approximation algorithm such as \LS{} on a sufficiently large set of ``top-valued'' elements $\TopElem$ only worsens its approximation by a factor of $\ApproxParamA$, since the optimal solution within the top-valued elements is an $\ApproxParamA$-approximation of the optimal solution in $\GroundSet$. Note also that if $\TopElem = \GroundSet$, then $\ApproxParam$ defaults to $(4+O(\ErrorB))^{-1}$.
	
	\begin{lemma}[Best $\Card$-size subset in top-valued elements \cite{SKOWRON201765}] \label{lemma:top elems opt value}
		Let $(\Funcf, \Card)$ be an instance of \SMCC{} where $\Funcf$ is $\PSep$-superseparable. Further, let $\ApproxParamA$ be a parameter such that $0 < \ApproxParamA < 1$, let $\TopElem$ be the set of top-$ \lceil \frac{\PSep\Card}{1-\ApproxParamA} + \Card \rceil $ elements $ \Elem \in \GroundSet $ by value $ \Funcf(\Elem) $, and let $\BestSetTopElem \subseteq \TopElem$ be a $\Card$-size subset that maximizes $\Funcf$. Then $\Funcf(\BestSetTopElem) \geq \ApproxParamA \OptValue$.
	\end{lemma}
	
	\subsection{\PGS{}} \label{sec:overview of PGS}
	\PGS{} (\PGSa{}), pseudocode in \cref{alg:PGS}, is the greedy-thresholding procedure called by \LSTPGSa{} to find an expected $(1-\frac{1}{e}-\Error)$-approximate solution in $O( \frac{1}{\Error^2} \log(\frac{\PSep \Card}{\Error}) )$ adaptive rounds and $O( \GroundSetSize + \frac{\PSep\Card}{\Error^2} + \frac{1}{\Error^2}\log( \frac{\PSep \Card}{\Error} ) )$ expected queries for constant $\ApproxParam$. Theoretical guarantees are given in \cref{theorem:PGS results}, and formally proven in \cref{sec:PGS analysis}.
	
	\begin{theorem} \label{theorem:PGS results}
		Let $(\Funcf, \Card)$ be an instance of \SMCC{} where $\Funcf$ is $\PSep$-superseparable. Suppose \PGSa{} (\cref{alg:PGS}) is run such that $\OptLB \leq \OptValue \leq \frac{\OptLB}{\ApproxParam}$ and $0 < \Error < 1-\frac{1}{e}$. Then, with probability $1-\frac{\ApproxParam\Error}{6 \PSep \Card}$, \PGSa{} achieves:
		\begin{itemize}
			\item a solution $\SetS$ satisfying $|\SetS| \leq \Card$ and $\mathbb{E}[ \Funcf(\SetS) ] \geq ( 1-\frac{1}{e}-\Error ) \OptValue$,
			\item an adaptive complexity of $ O( \frac{\log( \ApproxParam^{-1} )}{\Error^2} \log( \frac{\PSep \Card \log( \ApproxParam^{-1} ) }{\ApproxParam \Error}) ) $, and
			\item an expected query complexity of $ O( \GroundSetSize + \frac{(1-\ApproxParam)\PSep \Card}{\Error^2} + \frac{\log(\ApproxParam^{-1})}{\Error^2} \log( \frac{\PSep \Card}{\Error} ) ) $.
		\end{itemize}
	\end{theorem}
	
	We give the description of \PGSa{} and then explain how it achieves its adaptive complexity bounds with reference to \cref{lemma:expected gain of sample}. This is the key lemma showing that, for $\PSep$-superseparable \SMCC{}, sampling uniformly-at-random from sufficiently many high-value elements does not decrease the expected marginal gain of the remaining elements too much, giving a final sample with good expected marginal gain. We state and prove \cref{lemma:expected gain of sample} along with the preceding \cref{lemma:p-superseparable subset,lemma:average gain} in \cref{sec:PGS analysis}.
	
	\paragraph{Description of \PGS{}.} \PGSa{} takes in its input the values $ \OptLB $ and $ \ApproxParam $ such that $ \OptLB \leq \OptValue \leq \frac{\OptLB}{\ApproxParam} $. In \cref{line:PGS initialisation}, \PGSa{} initializes the solution $ \SetS_{0} \gets \varnothing $, the threshold $ \Threshold_{0} \gets \frac{\OptLB}{\ApproxParam \Card} \geq \frac{\OptValue}{\Card} $, and $ \MinSizeForSampleMax \gets \frac{3 \PSep \Card}{\ApproxParam \Error / 2} + \Card - 1$, which is the maximum number of elements that can be passed to \TBSa{} over the entire run of \PGSa{} (see \cref{claim:num good elems upper-bound}). For each element $\Elem \in \GroundSet$, \cref{line:PGS assign value} assigns $\ElemGain_\Elem \gets \Funcf(\Elem)$; these are used to build $\GoodElems_{\IterIndex}$ in \cref{line:PGS G update}.
	
	After the initialisation steps, \PGSa{} performs the steps below in each iteration,~$\IterIndex$, of the \cref{line:PGS while} loop. \PGSa{} differs from \PGBa{} \cite{chen2021best} in steps \ref{item:PGS construct good elems} and \ref{item:PGS sample or TBS}.
	
	\begin{enumerate}
		\item Assigns the threshold $\Threshold_\IterIndex$ by geometrically diminishing the previous threshold $\Threshold_{\IterIndex-1}$ (\cref{line:PGS tau update}).
		\item\label{item:PGS construct good elems} Constructs the set $\GoodElems_\IterIndex$ of elements $\Elem \in \GroundSet \setminus \SetS_{\IterIndex-1}$ with $\Funcf(\Elem) = \ElemGain_\Elem \geq \Threshold_{\IterIndex}$ (\cref{line:PGS G update}).
		\item\label{item:PGS sample or TBS} If $ |\GoodElems_{\IterIndex}| \geq \MinSizeForSample_{\IterIndex} $, \PGSa{} uniformly-at-random samples a set $\SetT_{\IterIndex}$ of size $ \Card - |\SetS_{\IterIndex-1}|$ from $\GoodElems_\IterIndex$ (\cref{line:PGS T sample}). Otherwise, if $|\GoodElems_{\IterIndex}| < \MinSizeForSample_{\IterIndex}$, \PGSa{} runs \TBSa{} on $\GoodElems_{\IterIndex}$ to obtain $\SetT_{\IterIndex}$ (\cref{line:PGS thresholding}). Either way, the expected marginal gain $ \Funcf(\SetT_{\IterIndex} \mid \SetS_{\IterIndex-1}) $ is approximately $ |\SetT_{\IterIndex}| \Threshold_{\IterIndex} $. This step is crucial to bounding the adaptive complexity as explained below.
		\item Produces $\SetS_{\IterIndex}$ by adding the set of new elements $\SetT_{\IterIndex}$ to $\SetS_{\IterIndex-1}$ (\cref{line:PGS solution update}).
	\end{enumerate}
	
	The \cref{line:PGS while} loop breaks if $\SetS_{\IterIndex}$ satisfies the cardinality constraint or if the threshold is too small to add elements with significant marginal gain.
	
	\paragraph{Bounding the Adaptive Complexity via $\PSep$-Superseparability.}
	The key idea behind the $O( \frac{1}{\Error^2} \log(\frac{\PSep \Card}{\Error}) )$ adaptive complexity of \PGSa{} is to bound the adaptive complexity of each iteration $\IterIndex$ in two cases for $|\GoodElems_{\IterIndex}|$. 
	
	When $|\GoodElems_{\IterIndex}| \geq \MinSizeForSample_{\IterIndex}$, \emph{no adaptive rounds} are needed to sample $ \SetT_{\IterIndex} \subseteq \GoodElems_{\IterIndex} $ in \cref{line:PGS T sample}. This is because \cref{lemma:expected gain of sample} ensures that, when $\Funcf$ is $\PSep$-superseparable and $|\GoodElems_{\IterIndex}| \geq \MinSizeForSample_{\IterIndex}$, $\SetT_{\IterIndex}$ has good expected marginal gain $ \Funcf(\SetT_{\IterIndex} \mid \SetS_{\IterIndex-1}) $.
	
	Otherwise, when $ |\GoodElems_{\IterIndex}| < \MinSizeForSample_{\IterIndex} $ (where $\MinSizeForSample_{\IterIndex} \leq \MinSizeForSampleMax = O(\frac{\PSep \Card}{\Error}) $ by \cref{claim:num good elems upper-bound}), running \TBSa{} on $\GoodElems_{\IterIndex}$ has a bounded adaptive complexity of $O( \frac{1}{\Error} \log(\frac{\PSep \Card}{\Error}) )$ (\cref{claim:PGS adaptive complexity iter} for constant $\ApproxParam$).
	
	The overall adaptive complexity of \PGSa{} follows since the case $|\GoodElems_{\IterIndex}| < \MinSizeForSample_{\IterIndex}$ (in which \TBSa{} is called) may occur in every \cref{line:PGS while} iteration, and the number of such iterations is bounded by $O(\log_{1-\Error}(\ApproxParam)) $ (\cref{claim:num PGS iters}), which is $O(\frac{1}{\Error})$ for constant $\ApproxParam$.
	
	\begin{algorithm}
		\caption{} \label{alg:PGS}
		\begin{algorithmic}[1]
			\Procedure{\PGS}{$\Funcf, \GroundSet, \PSep, \Card, \ApproxParam, \OptLB, \Error$}
			\Input{value oracle $ \Funcf : 2^\GroundSet \rightarrow \mathbb{R}_{\geq 0} $, ground set $ \GroundSet $, value $ \PSep $ such that $ \Funcf $ is $ \PSep $-superseparable, cardinality constraint $ \Card $, initial approximation factor $ \ApproxParam $, value $ \OptLB $ such that $ \OptLB \leq \OptValue \leq \frac{\OptLB}{\ApproxParam} $, approximation error $ \Error $}
			\Output{set $ \SetS \subseteq \GroundSet $ satisfying $ \mathbb{E}[\Funcf(\SetS)] \geq \left(1-\frac{1}{e}-\Error \right) \OptValue $}
			\State $ \IterIndex \gets 0,\, \SetS_0 \gets \varnothing,\, \Threshold_0 \gets \frac{\OptLB}{\ApproxParam \Card},\, \MinSizeForSampleMax \gets \frac{3 \PSep \Card}{\ApproxParam \Error / 2} + \Card - 1, \ProbFail \gets \Big(\log_{1-\Error}\left(\frac{\ApproxParam}{3}\right)\Big)^{-1} $ \label{line:PGS initialisation}
			\For{$ \Elem \in \GroundSet $}\label{line:PGS assign value loop}
			\State $ \ElemGain_\Elem \gets \Funcf( \Elem ) $ \label{line:PGS assign value}
			\EndFor
			\While{$ |\SetS_{\IterIndex}| < \Card $ \textnormal{and} $ \Threshold_{\IterIndex} \geq \frac{\OptLB}{(1-\Error) 3 \Card} $} \label{line:PGS while}
			\State $ \IterIndex \gets \IterIndex + 1 $ \label{line:PGS i update}
			\State $ \Threshold_{\IterIndex} \gets (1-\Error)^{\IterIndex} \frac{\OptLB}{\ApproxParam \Card} $ \label{line:PGS tau update}
			\State $ \GoodElems_{\IterIndex} \gets \{ \Elem \in \GroundSet \setminus \SetS_{\IterIndex-1} : \ElemGain_\Elem \geq \Threshold_{\IterIndex} \} $ \label{line:PGS G update}
			\State $ \MinSizeForSample_{\IterIndex} \gets \frac{\PSep \Card}{(1-\Error)^\IterIndex \Error / 2} + \Card - |\SetS_{\IterIndex-1}| - 1 $ \label{line:PGS m update}
			\If{$ |\GoodElems_{\IterIndex}| \geq \MinSizeForSample_{\IterIndex} $}\label{line:PGS if then sample}
			\State $ \SetT_{\IterIndex} \gets $ uniform-at-random sample of $ \Card - |\SetS_{\IterIndex-1}| $ elements from $ \GoodElems_{\IterIndex} $ \label{line:PGS T sample}
			\Else
			\State $ \SetT_{\IterIndex} \gets \textsc{\TBS}(\Funcf(\SetS_{\IterIndex-1} \cup \cdot \, ), \GoodElems_{\IterIndex}, \min\{\MinSizeForSampleMax, \GroundSetSize\}, \Card - |\SetS_{\IterIndex-1}|, \frac{\Error}{3}, \Prob, \Threshold_{\IterIndex}) $ \label{line:PGS thresholding}
			\EndIf
			\State $ \SetS_{\IterIndex} \gets \SetS_{\IterIndex-1} \cup \SetT_{\IterIndex} $ \label{line:PGS solution update}
			\EndWhile
			\State \Return{$ \SetS_{\IterIndex} $} \label{line:PGS return solution}
			\EndProcedure
		\end{algorithmic}
	\end{algorithm}
	
	\subsection{Analysis of \PGS{}} \label{sec:PGS analysis}	
	\paragraph{$\PSep$-Superseparability Lemmas for \PGSa{}.}
	Here, we prove \cref{lemma:p-superseparable subset}, which in turn is used to prove \cref{lemma:average gain}. Crucially, \cref{lemma:average gain} is used in the proof of \cref{lemma:expected gain of sample}.
	
	Notice that \cref{lemma:p-superseparable subset} essentially states that the properties of $\Funcf$ imply that it is also $\PSep$-superseparable  over an arbitrary subset $\GroundSubset\subseteq\GroundSet$, even for an arbitrary choice of $\SetA\subseteq\GroundSet$ in the inequality.
	
	\begin{lemma}\label{lemma:p-superseparable subset}
		If $ \Funcf $ is non-negative, submodular, and $ \PSep $-superseparable over ground set $ \GroundSet $, then $ \forall \GroundSubset\subseteq \GroundSet $ and $ \forall\SetA \subseteq \GroundSet $:
		\begin{align}
			\sum_{\Elem\in\GroundSubset} \Funcf(\Elem\mid \SetA) \geq \sum_{\Elem\in\GroundSubset} \Funcf(\Elem) - \PSep \Funcf(\SetA). \notag 
		\end{align}
	\end{lemma}
	
	\begin{proof}
		Begin with the definition of $ \PSep $-superseparability, which holds $ \forall \SetA \subseteq \GroundSet $, to derive \cref{lemma:p-superseparable subset}.
		\begin{align}
			\sum_{\Elem\in\GroundSet} \Funcf(\Elem \mid \SetA) &\geq \sum_{x\in\GroundSet} \Funcf(\Elem) - \PSep \Funcf(\SetA), &\text{$\PSep$-superseparability} \notag \\
			\sum_{\Elem\in\GroundSubset} \Funcf(\Elem\mid \SetA)  &\geq \sum_{x\in\GroundSet} \Funcf(\Elem) - \sum_{\Elem\in \GroundSet \setminus \GroundSubset } \Funcf( \Elem \mid \SetA) - \PSep \Funcf(\SetA) \notag \\
			&\geq \sum_{x\in\GroundSet} \Funcf(\Elem) - \sum_{\Elem\in \GroundSet \setminus \GroundSubset } \Funcf( \Elem ) - \PSep \Funcf(\SetA)\,, &\begin{aligned}
				&&\text{submodularity and} \\
				&&\text{non-negativity}
			\end{aligned} \notag \\
			\sum_{\Elem\in\GroundSubset} \Funcf(\Elem\mid \SetA) &\geq \sum_{x\in\GroundSubset} \Funcf(\Elem) - \PSep \Funcf(\SetA)\,. \notag
		\end{align} 
	\end{proof}
	
	\begin{lemma}\label{lemma:average gain}		
		Suppose $ \Funcf $ is non-negative, submodular, and $ \PSep $-superseparable over ground set $ \GroundSet $; and suppose for some $ \ErrorA > 0 $, $ \GroundSubset \subseteq \GroundSet $, and $\SetS\subseteq\GroundSet$, it holds that
		\begin{align}
			\ErrorA \sum_{\Elem \in \GroundSubset} \Funcf(\Elem) \geq \PSep \Funcf(\SetS)\,.\label{eqn:average gain condition}
		\end{align}
		It then follows that
		\begin{align}
			\sum_{\Elem \in \GroundSubset} \Funcf( \Elem \mid \SetA) \geq (1-\ErrorA) \sum_{\Elem \in \GroundSubset} \Funcf( \Elem )\,.\label{eqn:average gain}
		\end{align}
	\end{lemma}
	
	\begin{proof}
		Function~$ \Funcf $ satisfies the conditions for \cref{lemma:p-superseparable subset}, so from there we derive inequality~\eqref{eqn:average gain}.
		\begin{align}
			\sum_{\Elem\in\GroundSubset} \Funcf(\Elem\mid \SetA) &\geq \sum_{x\in\GroundSubset} \Funcf(\Elem) - \PSep \Funcf(\SetA) &\text{\cref{lemma:p-superseparable subset}} \notag \\
			&\geq \sum_{\Elem\in\GroundSubset} \Funcf(\Elem) - \ErrorA \sum_{\Elem \in \GroundSubset} \Funcf(\Elem) &\text{Inequality~\eqref{eqn:average gain condition}} \notag \\
			&= (1 - \ErrorA) \sum_{\Elem \in \GroundSubset} \Funcf(\Elem)\,. \notag
		\end{align}
	\end{proof}
	
	\begin{lemma}\label{lemma:expected gain of sample}
		Suppose \PGSa{} (\cref{alg:PGS}) is run such that $\OptLB \leq \OptValue \leq \frac{\OptLB}{\ApproxParam}$. Further, suppose that in some iteration $ \SampleIter $, \cref{line:PGS T sample} is executed so that $ \SetT_\SampleIter $ is assigned an ordered, uniform-at-random sample of $ \Card - |\SetS_{\SampleIter-1}| $ elements from $ \GoodElems_\SampleIter $ (without replacement). Then $ \SetT_\SampleIter $ satisfies
		\begin{align}
			\mathbb{E}[ \Funcf(\SetT_\SampleIter \mid \SetS_{\SampleIter-1} ) ] \geq |\SetT_\SampleIter| \left(1-\frac{\Error}{2}\right)\Threshold_\SampleIter\,. \notag
		\end{align}
	\end{lemma}
	
	Let $ \SetT_{\SampleIter, 0} = \varnothing $ and let $ \SetT_{\SampleIter, \SampleIndex} = \{ \SampleElem_1, \dots, \SampleElem_\SampleIndex \} $. That is, $ \SetT_{\SampleIter, \SampleIndex} $ is a partial sample of $ \SampleIndex $ elements from $ \GoodElems_\SampleIter $ (without replacement).
	
	We will first prove \cref{claim:good elements condition} below. This claim means that for every $ \SampleIndex \in \{ 1, \dots, |\SetT_\SampleIter| \} $, the set of elements $ \GoodElems_\SampleIter \setminus \SetT_{\SampleIter, \SampleIndex-1} $ satisfies inequality~\eqref{eqn:average gain condition} due to the success of the \cref{line:PGS if then sample} if-condition $ |\GoodElems_\SampleIter| \geq \MinSizeForSample_\SampleIter $, where inequality~\eqref{eqn:average gain condition} is a condition for \cref{lemma:average gain}. Function~$ \Funcf $ also satisfies the other conditions of \cref{lemma:average gain} by assumption. Then, by invoking \cref{lemma:average gain}, we will show that each sampled element $ \SampleElem_\SampleIndex $ has expected marginal gain $ \mathbb{E}[\Funcf(\SampleElem_{\SampleIndex} \mid \SetS_{\SampleIter-1} \cup \SetT_{\SampleIter, \SampleIndex-1})] \geq (1-\frac{\Error}{2}) \Threshold_\SampleIter $, where $ \SetS_{\SampleIter-1} \cup \SetT_{\SampleIter, \SampleIndex-1} $ is the previous solution combined with the previous partial sample. Finally, \cref{lemma:expected gain of sample} will follow by linearity of expectation.
	
	\begin{proof}
		We begin by proving \cref{claim:good elements condition}, which is needed to invoke \cref{lemma:average gain} later. 
		
		\begin{claim}\label{claim:good elements condition}
			For every $ \SampleIndex \in \{ 1, \dots, |\SetT_\SampleIter| \} $, $ \GoodElems_\SampleIter \setminus \SetT_{\SampleIter, \SampleIndex-1} $ satisfies
			\begin{align}
				\frac{\Error}{2}\cdot \sum_{\SampleElem_\SampleIndex \in \GoodElems_\SampleIter \setminus \SetT_{\SampleIter, \SampleIndex-1}} \Funcf( \SampleElem_\SampleIndex ) \geq \PSep \Funcf( \SetS_{\SampleIter-1} \cup \SetT_{\SampleIter, \SampleIndex-1} )\,. \notag
			\end{align}
		\end{claim}
		
		\begin{proof}
			Since we assume \cref{line:PGS T sample} is executed, the \cref{line:PGS if then sample} if-condition must have succeeded, meaning $ |\GoodElems_\SampleIter| \geq \MinSizeForSample_{\SampleIter} $. Then, for every $ \SampleIndex \in \{ 1, \dots, |\SetT_\SampleIter| \} $, we lower-bound $ |\GoodElems_\SampleIter \setminus \SetT_{\SampleIter, \SampleIndex-1}| $ as follows:
			\begin{align}
				|\GoodElems_\SampleIter \setminus \SetT_{\SampleIter, \SampleIndex-1}| &\geq \MinSizeForSample_{\SampleIter} - |\SetT_{\SampleIter, \SampleIndex-1}| \notag\\
				&\geq \MinSizeForSample_{\SampleIter} - (\Card - |\SetS_{\SampleIter-1}| - 1) \notag \\
				&= \frac{\PSep \Card}{(1-\Error)^\SampleIter \Error / 2}\,. &\text{via value of $ \MinSizeForSample_{\SampleIter} $ (\cref{line:PGS m update})} \label{eqn:G minus T lower-bound}
			\end{align}
			Then for every $ \SampleIndex \in \{ 1, \dots, |\SetT_\SampleIter| \} $, we prove the required inequality.
			\begin{align}
				\frac{\Error}{2} \cdot \sum_{\SampleElem_\SampleIndex \in \GoodElems_\SampleIter \setminus \SetT_{\SampleIter, \SampleIndex-1}} \Funcf( \SampleElem_\SampleIndex) &\geq \frac{\Error}{2} \cdot\sum_{\SampleElem_\SampleIndex \in \GoodElems_\SampleIter \setminus \SetT_{\SampleIter, \SampleIndex-1}} \Threshold_\SampleIter &\text{\cref{line:PGS assign value,line:PGS G update}} \notag \\
				&= \frac{\Error}{2} \cdot|\GoodElems_\SampleIter \setminus \SetT_{\SampleIter, \SampleIndex-1}| \Threshold_\SampleIter \notag \\
				&\geq \frac{\Error}{2} \frac{\PSep \Card}{(1-\Error)^\SampleIter \Error / 2} \Threshold_\SampleIter &\text{Inequality~\eqref{eqn:G minus T lower-bound} }  \notag \\
				&= \frac{\PSep \Card}{(1-\Error)^\SampleIter} (1-\Error)^\SampleIter \frac{\OptLB}{\ApproxParam \Card} &\text{value of $ \Threshold_\SampleIter $ (line~\ref{line:PGS tau update})} \notag \\
				&= \PSep \frac{\OptLB}{\ApproxParam} \notag \\
				&\geq \PSep \, \OptValue &\text{$ \frac{\OptLB}{\ApproxParam} \geq \OptValue $ on input to \PGSa{}} \notag \\
				&\geq \PSep \Funcf( \SetS_{\SampleIter-1} \cup \SetT_{\SampleIter, \SampleIndex-1} )\,. &\text{$\OptValue \geq \Funcf( \SetS_{\SampleIter-1} \cup \SetT_{\SampleIter, \SampleIndex-1} ) $} \notag
			\end{align} 
		\end{proof}
		
		Now by \cref{claim:good elements condition}, we have that for every $ \SampleIndex \in \{ 1, \dots, |\SetT_\SampleIter| \} $, inequality~\eqref{eqn:average gain condition} holds when $ \ErrorA = \frac{\Error}{2} $, $ \GroundSubset = \GoodElems_{\SampleIter} \setminus \SetT_{\SampleIter, \SampleIndex-1} $, and $ \SetS = \SetS_{\SampleIter-1} \cup \SetT_{\SampleIter, \SampleIndex-1} $. Recall that \cref{eqn:average gain condition} is a condition for \cref{lemma:average gain}. Further, $ \Funcf $ is assumed to be non-negative, monotone, submodular, and $ \PSep $-superseparable. Therefore, the conditions for \cref{lemma:average gain} hold for every $ \SampleIndex \in \{ 1, \dots, |\SetT_\SampleIter| \} $; we invoke this lemma in the proof of \cref{claim:exp marginal gain of element} below. Recall that $ \SetS_{\SampleIter-1} \cup \SetT_{\SampleIter, \SampleIndex-1} $ is the previous solution combined with a partial sample of $ \SampleIndex - 1 $ elements.
		\begin{claim}\label{claim:exp marginal gain of element}
			For every $ \SampleIndex \in \{ 1, \dots, |\SetT_\SampleIter| \} $, we have $ \mathbb{E}[\Funcf(\SampleElem_{\SampleIndex} \mid \SetS_{\SampleIter-1} \cup \SetT_{\SampleIter, \SampleIndex-1} )] \geq (1-\frac{\Error}{2}) \Threshold_\SampleIter $.
		\end{claim}
		
		\begin{proof}
			\begin{align}
				&\mathbb{E}[\Funcf(\SampleElem_{\SampleIndex} \mid \SetS_{\SampleIter-1} \cup \SetT_{\SampleIter, \SampleIndex-1} )] \notag \\
				&= \frac{1}{|\GoodElems_\SampleIter \setminus \SetT_{\SampleIter, \SampleIndex-1}|} \sum_{\SampleElem_\SampleIndex \in \GoodElems_\SampleIter \setminus \SetT_{\SampleIter, \SampleIndex-1} } \Funcf(\SampleElem_{\SampleIndex} \mid \SetS_{\SampleIter-1} \cup \SetT_{\SampleIter, \SampleIndex-1} ) \notag \\
				&\geq \frac{1}{|\GoodElems_\SampleIter \setminus \SetT_{\SampleIter, \SampleIndex-1}|} \left( 1-\frac{\Error}{2} \right) \sum_{\SampleElem_\SampleIndex \in \GoodElems_\SampleIter \setminus \SetT_{\SampleIter, \SampleIndex-1} } \Funcf(\SampleElem_{\SampleIndex}) &\text{\cref{lemma:average gain}} \notag \\
				&\geq \frac{1}{|\GoodElems_\SampleIter \setminus \SetT_{\SampleIter, \SampleIndex-1}|} \left(1-\frac{\Error}{2}\right) \sum_{\SampleElem_\SampleIndex \in \GoodElems_\SampleIter \setminus \SetT_{\SampleIter, \SampleIndex-1}} \Threshold_\SampleIter &\text{\cref{line:PGS assign value,line:PGS G update}} \notag \\
				&= \frac{1}{|\GoodElems_\SampleIter \setminus \SetT_{\SampleIter, \SampleIndex-1}|} \left(1-\frac{\Error}{2}\right) |\GoodElems_\SampleIter \setminus \SetT_{\SampleIter, \SampleIndex-1}| \Threshold_\SampleIter \notag \\
				&= \left(1-\frac{\Error}{2}\right) \Threshold_\SampleIter. \notag
			\end{align} 
		\end{proof}
		
		Finally, we show that $ \mathbb{E}[ \Funcf(\SetT_\SampleIter \mid \SetS_{\SampleIter-1} ) ] \geq |\SetT_\SampleIter| (1-\frac{\Error}{2})\Threshold_\SampleIter $. Recall that $ \SetT_{\SampleIter} $ is an ordered, uniform-at-random sample of elements from $ \GoodElems_\SampleIter $.
		\begin{align}
			\mathbb{E}[ \Funcf(\SetT_\SampleIter \mid \SetS_{\SampleIter-1} ) ] &= \mathbb{E}\left[ \sum_{\SampleIndex = 1}^{|\SetT_\SampleIter|} \Funcf( \SampleElem_{\SampleIndex} \mid \SetS_{\SampleIter-1} \cup \SetT_{\SampleIter, \SampleIndex-1} ) \right] &\text{telescoping series} \notag \\
			&= \sum_{\SampleIndex = 1}^{|\SetT_\SampleIter|} \mathbb{E}[ \Funcf( \SampleElem_{\SampleIndex} \mid \SetS_{\SampleIter-1} \cup \SetT_{\SampleIter, \SampleIndex-1} ) ] &\text{linearity of expectation} \notag \\
			&\geq \sum_{\SampleIndex = 1}^{|\SetT_\SampleIter|} \left( 1 - \frac{\Error}{2} \right) \Threshold_\SampleIter &\text{\cref{claim:exp marginal gain of element}} \notag \\
			&= |\SetT_\SampleIter| \left(1-\frac{\Error}{2}\right)\Threshold_\SampleIter\,. \notag
		\end{align} 
	\end{proof}
	
	\paragraph{Main Proof of Approximation Factor of \PGSa{}.}
	The remaining analysis of the approximation factor of \PGSa{} is very similar to that of \PGBa{} \cite{chen2021best}, containing only a few minor modifications. Nevertheless, we give the full analysis for completeness.
	
	We first provide \cref{lemma:exp upper-bound,claim:error and e factor lower-bound,claim:error factor lower-bound}, which are used in the main approximation factor proof in \cref{lemma:PGS approx factor}. The proof of \cref{lemma:exp upper-bound} is given in Lemma 14 of Chen et al.~\cite{chen2021best}.
	
	\begin{lemma}[Upper-bound on exponential term \cite{chen2021best}] \label{lemma:exp upper-bound}
		For $0 \leq \Error \leq 1$, we have $e^{-\frac{(1-\Error/3)(1-\Error)}{1+\Error/3}} \leq \frac{1}{e}+\Error$. 
	\end{lemma}
	
	\begin{claim}\label{claim:error and e factor lower-bound}
		For all $ 0 \leq \Error \leq \frac{2}{3} $, it holds that $ 1-\frac{1}{(1-\Error)3} \geq \frac{2}{3}-\Error $.
	\end{claim}
	
	\begin{proof}
		The required inequality can be simplified to $ \Error ( \frac{2}{3} -\Error ) \geq 0$, which holds for all $ 0 \leq \Error \leq \frac{2}{3} $. 
	\end{proof}
	
	\begin{claim}\label{claim:error factor lower-bound}
		For all $ 0 \leq \Error \leq 1 $, it holds that $ 1-\frac{\Error}{2} \geq \frac{1-\Error/3}{1+\Error/3} $.
	\end{claim}
	
	\begin{proof}
		The required inequality can be simplified to $ \Error \geq \Error^2$, which holds for all $ 0 \leq \Error \leq 1 $. 
	\end{proof}
	
	\begin{lemma}\label{lemma:PGS approx factor}
		Let $(\Funcf, \Card)$ be an instance of \SMCC{} where $\Funcf$ is $\PSep$-superseparable. Suppose \PGSa{} (\cref{alg:PGS}) is run such that $\OptLB \leq \OptValue \leq \frac{\OptLB}{\ApproxParam}$ and $0 < \Error < 1-\frac{1}{e}$. Further, suppose that \PGSa{} terminates successfully, and let $\SampleIter$ be the index of the final iteration of the \cref{line:PGS while} loop. Then \PGSa{} returns a solution $\SetS_{\SampleIter}$ such that $|\SetS_{\SampleIter}| \leq \Card$ and
		\begin{align}
			\mathbb{E}[\Funcf(\SetS_{\SampleIter})] \geq \left(1-\frac{1}{e}-\Error\right) \OptValue\,. \notag
		\end{align}
	\end{lemma}
	
	\begin{proof}
		We prove the approximation factor of \PGSa{} in two cases: $|\SetS_{\SampleIter}| < \Card$ and $|\SetS_{\SampleIter}| = \Card$. Before this, we prove \cref{claim:marginal gain opt ub} below.
		
		\begin{claim} \label{claim:marginal gain opt ub}
			For each iteration $\IterIndex \geq 1$ of the \cref{line:PGS while} loop, if $|\SetS_{\IterIndex}| < \Card$, then for all optimal elements \mbox{$\OptElem \in \OptSet \setminus \SetS_{\IterIndex} \colon \Funcf(\OptElem \mid \SetS_{\IterIndex}) < \Threshold_{\IterIndex}$}.
		\end{claim}
		
		\begin{proof}
			If $|\SetS_{\IterIndex}| < \Card$, then $\SetS_{\IterIndex}$ had room to include new elements without violating the cardinality constraint. This means that, for each $\OptElem \in \OptSet \setminus \SetS_{\IterIndex}$, either
			\begin{itemize}
				\item $\OptElem \notin \GoodElems_{\IterIndex}$, which means $\Funcf(\OptElem) = \ElemGain_{\OptElem} < \Threshold_{\IterIndex}$ (\cref{line:PGS assign value,line:PGS G update}). Therefore, by submodularity and non-negativity, $\Funcf(\OptElem \mid \SetS_{\IterIndex}) \leq \Funcf(\OptElem) < \Threshold_{\IterIndex}$.
				\item $\OptElem \in \GoodElems_{\IterIndex}$. But due to $|\SetS_{\IterIndex}| < \Card$, the added set $\SetT_{\IterIndex}$ must have been assigned by \TBSa{} (\cref{line:PGS thresholding}) as the other assignment of $\SetT_{\IterIndex}$ (\cref{line:PGS T sample}) would have given $ |\SetS_{\IterIndex}| = \Card $. Then $\OptElem$ must have failed the threshold within \TBSa{}. That is, by \cref{theorem:TBS results}, $\Funcf(\OptElem \mid \SetS_{\IterIndex}) = \Funcf(\OptElem \mid \SetS_{\IterIndex-1} \cup \SetT_{\IterIndex}) = \Funcg(\OptElem \mid \SetT_{\IterIndex}) < \Threshold_{\IterIndex}$. 
			\end{itemize} 
		\end{proof}
		
		\paragraph*{Case $|\SetS_{\SampleIter}| < \Card$.}
		We prove the approximation factor in the case $|\SetS_{\SampleIter}| < \Card$. Note below that the final threshold $\Threshold_{\SampleIter}$ satisfies $\Threshold_{\SampleIter} < \frac{\OptLB}{(1-\Error) 3 \Card}$, as this is the only way to exit the \cref{line:PGS while} loop when $|\SetS_{\SampleIter}| < \Card$.
		\begin{align}
			\OptValue - \Funcf(\SetS_{\SampleIter}) &\leq \Funcf(\OptSet \cup \SetS_{\SampleIter}) - \Funcf(\SetS_{\SampleIter}) &\text{monotonicity} \notag \\
			&\leq \sum_{\OptElem\in\OptSet\setminus\SetS_{\SampleIter}} \Funcf(\OptElem \mid \SetS_{\SampleIter}) &\text{submodularity} \notag \\
			&\leq \Card \Threshold_{\SampleIter} &\text{\cref{claim:marginal gain opt ub}} \notag \\
			&< \Card \frac{\OptLB}{(1-\Error) 3 \Card} &\text{$\Threshold_{\SampleIter}$ fails the \cref{line:PGS while} condition} \notag \\
			&= \frac{\OptLB}{(1-\Error) 3} \notag \\
			&\leq \frac{\OptValue}{(1-\Error) 3}\,, &\text{$ \OptLB \leq \OptValue $ holds on input to \PGSa{}} \notag \\
			\Funcf(\SetS_{\SampleIter}) &\geq \left(1-\frac{1}{(1-\Error)3}\right)\OptValue \notag \\
			&\geq \left(\frac{2}{3} - \Error \right)\OptValue\,. &\text{\cref{claim:error and e factor lower-bound} and $0 < \Error < 1-\frac{1}{e}$} \notag
		\end{align}
		
		\paragraph*{Case $|\SetS_{\SampleIter}| = \Card$.} 
		First, we prove \cref{claim:threshold lower-bound}.
		
		\begin{claim} \label{claim:threshold lower-bound}
			For each $\IterIndex \in \{1, \dots, \SampleIter\}$, we have
			\begin{align}
				\Threshold_{\IterIndex} &\geq (1-\Error) \frac{1}{\Card} \left( \OptValue - \Funcf(\SetS_{\IterIndex-1}) \right). \notag
			\end{align}
		\end{claim}
		
		\begin{proof}
			First we lower-bound $\Threshold_{\IterIndex}$ for $\IterIndex = 1$.
			\begin{align}
				\Threshold_{1} &= (1-\Error)\frac{\OptLB}{\ApproxParam \Card} &\text{value of $\Threshold_{1}$ (\cref{line:PGS tau update})} \notag \\
				&\geq (1-\Error)\frac{1}{\Card}\OptValue &\text{$\frac{\OptLB}{\ApproxParam} \geq \OptValue$ holds on input to \PGSa{}} \notag \\
				&\geq (1-\Error)\frac{1}{\Card}(\OptValue - \Funcf(\SetS_{\IterIndex-1}))\,. &\text{non-negativity} \notag
			\end{align}
			Then, for every $\IterIndex \in \{2, \dots, \SampleIter \}$, we lower-bound $\Threshold_{\IterIndex}$. Note that $ |\SetS_{\IterIndex-1}| < \Card $ since $ \IterIndex - 1 < \SampleIter $.
			\begin{align}
				\Threshold_{\IterIndex} &= (1-\Error) \Threshold_{\IterIndex-1} &\text{\cref{line:PGS tau update}} \notag \\
				&\geq (1-\Error) \frac{1}{\Card}\sum_{\OptElem \in \OptSet\setminus \SetS_{\IterIndex-1}} \Funcf(\OptElem \mid \SetS_{\IterIndex-1}) &\text{\cref{claim:marginal gain opt ub} and $ |\SetS_{\IterIndex-1}| < \Card $} \notag\\
				&\geq (1-\Error) \frac{1}{\Card} \left( \Funcf(\OptSet\cup\SetS_{\IterIndex-1}) - \Funcf(\SetS_{\IterIndex-1}) \right) &\text{submodularity} \notag \\
				&\geq (1-\Error) \frac{1}{\Card} \left( \OptValue - \Funcf(\SetS_{\IterIndex-1}) \right)\,. &\text{monotonicity} \notag
			\end{align} 
		\end{proof}	
		
		Now, for every $\IterIndex \in \{1,\dots,\SampleIter\}$, we lower-bound $\Funcf(\SetT_{\IterIndex} \mid \SetS_{\IterIndex-1})$ assuming that every $\SetT_{\IterIndex}$ was assigned by \TBSa{} in \cref{line:PGS thresholding}.
		\begin{align}
			\Funcf(\SetT_{\IterIndex} \mid \SetS_{\IterIndex-1}) &\geq \frac{1-\Error/3}{1+\Error/3} \Threshold_{\IterIndex} |\SetT_{\IterIndex}| &\text{\cref{theorem:TBS results}} \notag \\
			&\geq \frac{(1-\Error/3)(1-\Error)}{1+\Error/3} \frac{|\SetT_{\IterIndex}|}{\Card} (\OptValue - \Funcf(\SetS_{\IterIndex-1}) )\,, & \label{eqn:marginal gain Ti}
		\end{align}
		where the latter follows from \cref{claim:threshold lower-bound}.
		Next, we lower-bound $\mathbb{E}[\Funcf(\SetT_{\SampleIter} \mid \SetS_{\SampleIter-1})]$ assuming that the final added set $\SetT_{\SampleIter}$ was sampled in \cref{line:PGS T sample}. (Note that no earlier set $\SetT_{\IterIndex}$ could have been sampled since the sampling always fills the solution with $\Card$ elements, causing the \cref{line:PGS while} loop to break.)
		\begin{align}
			\mathbb{E}[ \Funcf(\SetT_\SampleIter \mid \SetS_{\SampleIter-1} ) ] &\geq |\SetT_\SampleIter| \left(1-\frac{\Error}{2}\right)\Threshold_\SampleIter &\text{\cref{lemma:expected gain of sample}} \notag \\
			&\geq \left(1-\frac{\Error}{2}\right)(1-\Error) \frac{|\SetT_\SampleIter|}{\Card} (\OptValue - \Funcf(\SetS_{\SampleIter-1}) ) &\text{\cref{claim:threshold lower-bound}} \notag \\
			&\geq \frac{(1-\Error/3)(1-\Error)}{1+\Error/3} \frac{|\SetT_{\SampleIter}|}{\Card} (\OptValue - \Funcf(\SetS_{\SampleIter-1}) )\,, & \label{eqn:marginal gain sampled Tu}
		\end{align}
		where the latter follows from \cref{claim:error factor lower-bound}.

		Now, for every $\IterIndex\in\{1, \dots, \SampleIter \}$, we rearrange \cref{eqn:marginal gain Ti} to derive an upper-bound for $\OptValue - \Funcf(\SetS_{\IterIndex})$ assuming that every $\SetT_{\IterIndex}$ was assigned by \TBSa{} in \cref{line:PGS thresholding}.
		\begin{align}
			\Funcf(\SetT_{\IterIndex} \mid \SetS_{\IterIndex-1}) &\geq \frac{(1-\Error/3)(1-\Error)}{1+\Error/3} \frac{|\SetT_{\IterIndex}|}{\Card} (\OptValue - \Funcf(\SetS_{\IterIndex-1}) )\,, \notag \\
			\Funcf(\SetS_{\IterIndex}) - \Funcf(\SetS_{\IterIndex-1}) &\geq \frac{(1-\Error/3)(1-\Error)}{1+\Error/3} \frac{|\SetT_{\IterIndex}|}{\Card} (\OptValue - \Funcf(\SetS_{\IterIndex-1}) )\,, \notag \\
			\OptValue - \Funcf(\SetS_{\IterIndex}) &\leq \left(1-\frac{(1-\Error/3)(1-\Error)}{1+\Error/3} \frac{|\SetT_{\IterIndex}|}{\Card}\right) (\OptValue - \Funcf(\SetS_{\IterIndex-1}) )\,. \label{eqn:opt Si upper-bound}
		\end{align}
		
		We similarly rearrange Inequality~\eqref{eqn:marginal gain sampled Tu} to derive an upper-bound for $\OptValue - \mathbb{E}[\Funcf(\SetS_{\SampleIter})]$ assuming that $\SetT_{\SampleIter}$ was sampled in \cref{line:PGS T sample}. Note that we replace $\mathbb{E}[\Funcf(\SetS_{\SampleIter-1})]$ with $\Funcf(\SetS_{\SampleIter-1})$ since the expectation is with respect to the uniform-at-random sampling of the added set $\SetT_{\SampleIter}$, and the value of $\Funcf(\SetS_{\SampleIter-1})$ is constant once the sampling begins.
		\begin{align}
			\mathbb{E}[\Funcf(\SetT_{\SampleIter} \mid \SetS_{\SampleIter-1})] &\geq \frac{(1-\Error/3)(1-\Error)}{1+\Error/3} \frac{|\SetT_{\SampleIter}|}{\Card} (\OptValue - \Funcf(\SetS_{\SampleIter-1}) )\,, \notag \\
			\mathbb{E}[\Funcf(\SetS_{\SampleIter})] - \Funcf(\SetS_{\SampleIter-1}) &\geq \frac{(1-\Error/3)(1-\Error)}{1+\Error/3} \frac{|\SetT_{\SampleIter}|}{\Card} (\OptValue - \Funcf(\SetS_{\SampleIter-1}) )\,, \notag \\
			\OptValue - \mathbb{E}[\Funcf(\SetS_{\SampleIter})] &\leq \left(1-\frac{(1-\Error/3)(1-\Error)}{1+\Error/3} \frac{|\SetT_{\SampleIter}|}{\Card}\right) (\OptValue - \Funcf(\SetS_{\SampleIter-1}) )\,. \label{eqn:opt sampled Su upper-bound}
		\end{align}
		
		Finally, we prove the approximation factor in the case $|\SetS_{\SampleIter}| = \Card$. Assuming that $\SetT_{\SampleIter}$ was sampled in \cref{line:PGS T sample}, we begin with \cref{eqn:opt sampled Su upper-bound} and then chain \cref{eqn:opt Si upper-bound} on the right-hand side for each $\IterIndex \in \{\SampleIter-1, \dots, 1\}$. If we instead assume that every $\SetT_{\IterIndex}$ was assigned by \TBSa{} in \cref{line:PGS thresholding}, the approximation factor is derived in the same way except that we simply begin with \cref{eqn:opt Si upper-bound} instead of \cref{eqn:opt sampled Su upper-bound}.
		\begin{align}
			&\OptValue - \mathbb{E}[\Funcf(\SetS_{\SampleIter})] \notag \\
			&\leq \prod_{\IterIndex=1}^{\SampleIter}\left(1 - \frac{(1-\Error/3)(1-\Error)}{(1+\Error/3)} \frac{|\SetT_{\IterIndex}|}{\Card} \right) \cdot (\OptValue - \Funcf(\SetS_{0})) \notag \\
			&\leq \prod_{\IterIndex=1}^{\SampleIter}\left(1 - \frac{(1-\Error/3)(1-\Error)}{(1+\Error/3)} \frac{|\SetT_{\IterIndex}|}{\Card} \right) \cdot \OptValue &\text{non-negativity} \notag \\
			&\leq \prod_{\IterIndex=1}^{\SampleIter} \left( e^{ -\frac{(1-\Error/3)(1-\Error)}{(1+\Error/3)} \frac{|\SetT_{\IterIndex}|}{\Card} } \right) \cdot \OptValue \notag \\
			&= e^{ -\frac{(1-\Error/3)(1-\Error)}{(1+\Error/3)} \frac{1}{\Card} \sum_{\IterIndex=1}^{\SampleIter} |\SetT_{\IterIndex}| }  \cdot \OptValue \notag \\
			&= e^{ -\frac{(1-\Error/3)(1-\Error)}{(1+\Error/3)} } \cdot \OptValue &\text{$ \sum_{\IterIndex=1}^{\SampleIter} |\SetT_{\IterIndex}| = |\SetS_{\SampleIter}| = \Card $} \notag \\
			&\leq \left( \frac{1}{e} + \Error \right) \OptValue\,, &\text{\cref{lemma:exp upper-bound}} \notag \\
			\mathbb{E}[\Funcf(\SetS_{\SampleIter})]	&\geq \left(1 - \frac{1}{e} - \Error \right) \OptValue\,. \notag
		\end{align} 
	\end{proof}
	
	\paragraph{Preliminary Claims for \PGSa{}.}
	Here we prove Claims~\ref{claim:num PGS iters} and \ref{claim:num good elems upper-bound}, which are used in the proofs of the adaptive complexity (\cref{lemma:PGS adaptive complexity}), query complexity (\cref{lemma:PGS query complexity}), and success probability (\cref{lemma:PGS success probability}) of \PGSa{} (\cref{alg:PGS}).
	\begin{claim} \label{claim:num PGS iters}
		In \PGSa{} (\cref{alg:PGS}), there are at most $ \log_{1-\Error}( \frac{\ApproxParam}{3} ) $ iterations of the \cref{line:PGS while} loop. That is, the iteration index $\IterIndex \leq \log_{1-\Error}( \frac{\ApproxParam}{3} ) $.
	\end{claim}
	
	\begin{proof}
		Let $ \SampleIter $ be the index of the final iteration of the \cref{line:PGS while} loop, which must also be the number of loop iterations. To upper-bound $ \SampleIter $, observe that upon entering iteration $ \SampleIter $, $\Threshold_{\SampleIter-1} \geq \frac{\OptLB}{(1-\Error)3\Card}$ must hold due to the loop condition. Further, we have that $\Threshold_{\SampleIter-1} = (1-\Error)^{\SampleIter-1} \frac{\OptLB}{\ApproxParam \Card}$ as assigned in \cref{line:PGS tau update} of the previous iteration. Thus, solving for $\SampleIter$ in $(1-\Error)^{\SampleIter-1} \frac{\OptLB}{\ApproxParam \Card} \geq \frac{\OptLB}{(1-\Error) 3 \Card}$ gives $\SampleIter \leq \log_{1-\Error}( \frac{\ApproxParam}{3} )$. 
	\end{proof}
	
	\begin{claim} \label{claim:num good elems upper-bound}
		In \PGSa{} (\cref{alg:PGS}), for every iteration $ \IterIndex $ of the \cref{line:PGS while} loop, it holds that if \cref{line:PGS thresholding} is run so that \TBSa{} (\cref{alg:TBS}) is called with input set $ \GoodElems_{\IterIndex} $, then $ |\GoodElems_{\IterIndex}| < \MinSizeForSample_{\IterIndex} \leq \MinSizeForSampleMax $, where $\MinSizeForSampleMax = \frac{3 \PSep \Card}{\ApproxParam \Error / 2} + \Card - 1 $ as in \cref{line:PGS initialisation} of \PGSa{}.
	\end{claim}
	
	\begin{proof}
		For an arbitrary iteration $ \IterIndex $ of the \cref{line:PGS while} loop, in order for \cref{line:PGS thresholding} to run, $ |\GoodElems_{\IterIndex}| < \MinSizeForSample_{\IterIndex} $ must hold due to the \cref{line:PGS if then sample} if-condition. Below, we further show that $\MinSizeForSample_{\IterIndex} \leq \MinSizeForSampleMax$ by bounding $\IterIndex$, proving the claim.
		\begin{align}
			|\GoodElems_{\IterIndex}| &< \MinSizeForSample_{\IterIndex}, \notag \\
			&\leq \frac{\PSep \Card}{(1-\Error)^\IterIndex \Error / 2} + \Card - 1 &\text{value of $\MinSizeForSample_{\IterIndex}$ (\cref{line:PGS m update})} \notag \\
			&\leq \frac{\PSep \Card}{(1-\Error)^{\log_{1-\Error}\left( \frac{\ApproxParam}{3} \right)} \Error / 2} + \Card - 1 &\text{\cref{claim:num PGS iters}} \notag \\
			&= \frac{3 \PSep \Card}{\ApproxParam \Error / 2} + \Card - 1 \notag \\
			&= \MinSizeForSampleMax\,. &\text{\cref{line:PGS initialisation}} \notag
		\end{align} 
	\end{proof}	
	
	\paragraph{Adaptive Complexity of \PGSa{}.}
	\begin{lemma} \label{lemma:PGS adaptive complexity}
		The adaptive complexity of \PGSa{} (\cref{alg:PGS}) is $ O( \frac{\log( \ApproxParam^{-1} )}{\Error^2} \log( \frac{\PSep \Card \log( \ApproxParam^{-1} ) }{\ApproxParam \Error} ) ) $.
	\end{lemma}
	
	\begin{proof}		
		First, we prove \cref{claim:PGS adaptive complexity iter} below.
		\begin{claim}\label{claim:PGS adaptive complexity iter}
			In \PGSa{} (\cref{alg:PGS}), the adaptive complexity of each iteration of the \cref{line:PGS while} loop is $ O( \frac{1}{\Error}\log( \frac{\PSep \Card \log( \ApproxParam^{-1} ) }{\ApproxParam \Error} ) ) $.
		\end{claim}
		
		\begin{proof}
			Each loop iteration requires $ O( \frac{1}{\Error}\log (\frac{\MinSizeForSampleMax}{\ProbFail} ) ) $ adaptive rounds since running \TBSa{} (\cref{line:PGS thresholding}) uses this many adaptive rounds by \cref{theorem:TBS results}. Crucially, observe that if the sampling step in \cref{line:PGS T sample} is run instead of \cref{line:PGS thresholding}, it would require no adaptive rounds. By simplifying this bound, we prove \cref{claim:PGS adaptive complexity iter} below.
			\begin{align}
				&\text{Adaptive complexity of one iteration of \textsc{\PGSa}} \notag \\
				&= O\left( \frac{1}{\Error}\log\left( \frac{\MinSizeForSampleMax}{\ProbFail} \right) \right) \notag \\
				&= O\left( \frac{1}{\Error} \log\left( \frac{\PSep \Card}{\ApproxParam \Error \ProbFail} \right) \right) &\text{$\MinSizeForSampleMax = \frac{3 \PSep \Card}{\ApproxParam \Error / 2} + \Card - 1$ (\cref{line:PGS initialisation})} \notag \\
				&= O\left( \frac{1}{\Error}\log\left( \frac{\PSep \Card \log( \ApproxParam^{-1} ) }{\ApproxParam \Error} \right) \right)\,. &\text{$ \ProbFail = \left( \log_{1-\Error}\left(\frac{\ApproxParam}{3}\right) \right)^{-1} $ (\cref{line:PGS initialisation})} \notag
			\end{align} 
		\end{proof}
		
		To bound the overall adaptive complexity of \cref{alg:PGS}, there is initially 1 round of queries to assign each $\ElemGain_{\Elem} \gets \Funcf(\Elem)$ (\cref{line:PGS assign value loop,line:PGS assign value}). Then there are at most $ \log_{1-\Error}( \frac{\ApproxParam}{3} ) = O( \frac{\log( \ApproxParam^{-1} )}{\Error} ) $ sequential iterations of the \cref{line:PGS while} loop by \cref{claim:num PGS iters}, and each such iteration has adaptive complexity $O( \frac{1}{\Error}\log( \frac{\PSep \Card \log( \ApproxParam^{-1} ) }{\ApproxParam \Error} ) )$ by \cref{claim:PGS adaptive complexity iter}. Multiplying these bounds gives the required adaptive complexity. 
	\end{proof}
	
	\paragraph{Query Complexity of \PGSa{}.}
	\begin{lemma} \label{lemma:PGS query complexity}
		The expected query complexity of \PGSa{} (\cref{alg:PGS}) is $ O( \GroundSetSize + \frac{(1-\ApproxParam)\PSep \Card}{\Error^2} + \frac{\log(\ApproxParam^{-1})}{\Error^2} \log( \frac{\PSep \Card}{\Error} ) ) $.
	\end{lemma}
	
	\begin{proof}
		First, we prove \cref{claim:PGS query complexity iter} below.
		
		\begin{claim}\label{claim:PGS query complexity iter}
			In \PGSa{} (\cref{alg:PGS}), the expected query complexity of the $\IterIndex$th iteration of the \cref{line:PGS while} loop is $ O( \frac{\PSep \Card}{(1-\Error)^{\IterIndex} \Error} + \frac{1}{\Error} \log( \frac{\PSep \Card}{\Error} ) ) $.
		\end{claim}
		
		\begin{proof}
			Each iteration $\IterIndex$ of the \cref{line:PGS while} loop uses $ O(|\GoodElems_{\IterIndex}| + \frac{\log(|\GoodElems_{\IterIndex}|)}{\Error} ) $ queries in expectation since running \TBSa{} (\cref{line:PGS thresholding}) uses this many queries in expectation by \cref{lemma:TBS query complexity}. Note that if the sampling step in \cref{line:PGS T sample} is run instead of \cref{line:PGS thresholding}, it would require no additional queries. By simplifying this bound, we prove \cref{claim:PGS query complexity iter} below.
			\begin{align}
				&\mathbb{E}[\text{Query complexity of $ \IterIndex $th iteration of \textsc{\PGSa}}] \notag \\
				&= O\left(|\GoodElems_{\IterIndex}| + \frac{\log(|\GoodElems_{\IterIndex}|)}{\Error} \right) \notag \\
				&= O\left(\MinSizeForSample_{\IterIndex} + \frac{\log(\MinSizeForSampleMax)}{\Error} \right) &\text{\cref{claim:num good elems upper-bound}} \notag \\
				&= O\left(\frac{\PSep \Card}{(1-\Error)^{\IterIndex} \Error} + \frac{1}{\Error}\log\left( \frac{\PSep \Card}{\Error} \right) \right)\,. &\text{\cref{line:PGS initialisation,line:PGS m update}} \notag
			\end{align} 
		\end{proof}
		
		To bound the overall expected query complexity of \PGSa{}, there are initially $\GroundSetSize$ queries to assign each $\ElemGain_{\Elem} \gets \Funcf(\Elem)$ (\cref{line:PGS assign value loop,line:PGS assign value}). Then we sum the expected query complexities over all iterations $\IterIndex$ of the \cref{line:PGS while} loop. Each iteration $\IterIndex$ has expected query complexity $O( \frac{\PSep \Card}{(1-\Error)^{\IterIndex} \Error} + \frac{1}{\Error}\log( \frac{\PSep \Card}{\Error} ) )$ by \cref{claim:PGS query complexity iter}. Further, there are at most $ \log_{1-\Error}( \frac{\ApproxParam}{3} ) $ iterations $\IterIndex$ by \cref{claim:num PGS iters}. Thus, we prove \cref{lemma:PGS query complexity} below, where \cref{eqn:query geometric series} follows from the geometric series.
		\begin{align}
			&\mathbb{E}[\text{Query complexity of \textsc{\PGSa}}] \notag \\
			&\leq \GroundSetSize + \sum_{\IterIndex=1}^{ \left\lceil \log_{1-\Error}\left( \frac{\ApproxParam}{3} \right) \right\rceil } O\left( \frac{\PSep \Card}{(1-\Error)^{\IterIndex} \Error} + \frac{1}{\Error}\log\left( \frac{\PSep \Card}{\Error} \right) \right) \notag \\
			&\leq O\left( \GroundSetSize + \frac{\PSep \Card}{\Error} \left( \frac{1 - (1-\Error)^{  \left\lceil \log_{1-\Error}\left( \frac{\ApproxParam}{3} \right) \right\rceil - 1 }}{1 - (1-\Error)} \right) + \frac{\log_{1-\Error}\left( \ApproxParam \right)}{\Error}\log\left( \frac{\PSep \Card}{\Error} \right) \right) \label{eqn:query geometric series} \\ 
			&\leq O\left( \GroundSetSize + \frac{\PSep \Card}{\Error} \left( \frac{1 - (1-\Error)^{\log_{1-\Error}\left( \frac{\ApproxParam}{3} \right)}}{\Error} \right) + \frac{\log(\ApproxParam^{-1})}{\Error^2} \log\left( \frac{\PSep \Card}{\Error} \right) \right) \notag \\
			&= O\left( \GroundSetSize + \frac{(1-\ApproxParam)\PSep \Card}{\Error^2} + \frac{\log(\ApproxParam^{-1})}{\Error^2} \log\left( \frac{\PSep \Card}{\Error} \right) \right)\,. \notag
		\end{align} 
	\end{proof}
	
	\paragraph{Success Probability of \PGSa{}.}
	\begin{lemma} \label{lemma:PGS success probability}
		\PGSa{} (\cref{alg:PGS}) terminates successfully with probability $1-\frac{\ApproxParam\Error}{6 \PSep \Card}$.
	\end{lemma}
	
	\begin{proof}
		We define a `\TBSFailure{}' as the event where \TBSa{} (\cref{alg:TBS}) returns \textit{failure}. Further, we define a `\PGSFailure{}' as the event where, over an execution of \PGSa{}, a call to \TBSa{} in \cref{line:PGS thresholding} results in a \TBSFailure{}. Each call to \TBSa{} passes $\GoodElems_{\IterIndex}$, $\MinSizeForSampleMax$, and $\ProbFail$, and $|\GoodElems_{\IterIndex}| < \MinSizeForSampleMax$ always holds by \cref{claim:num good elems upper-bound}. Hence, $\Pr[\textnormal{\TBSFailure}] \leq \frac{\ProbFail}{\MinSizeForSampleMax}$ holds by \cref{lemma:TBS success probability}.
		
		The call to \TBSa{} could be made in every iteration of the \cref{line:PGS while} loop, and there are at most $\log_{1-\Error}( \frac{\ApproxParam}{3} ) $ such iterations by \cref{claim:num PGS iters}. Thus, we bound $\Pr[\textnormal{\PGSFailure}]$ below.
		\begin{align}
			\Pr[\textnormal{\PGSFailure}] &\leq \log_{1-\Error}\left( \frac{\ApproxParam}{3} \right) \Pr[\textnormal{\TBSFailure}] \notag \\
			&\leq \log_{1-\Error}\left( \frac{\ApproxParam}{3} \right) \frac{\ProbFail}{\MinSizeForSampleMax} &\text{\cref{lemma:TBS success probability}} \notag \\
			&\leq \log_{1-\Error}\left( \frac{\ApproxParam}{3} \right) \left(\log_{1-\Error}\left( \frac{\ApproxParam}{e} \right)\right)^{-1} \frac{\ApproxParam \Error/2}{3 \PSep \Card} &\text{\cref{line:PGS initialisation}} \notag \\
			&= \frac{\ApproxParam \Error}{6 \PSep \Card}\,. \notag
		\end{align} 
	\end{proof}
	
	\paragraph{Performance Guarantees of \LSTPGSa{} for General \SMCC{}.}
	Here, we give the performance guarantees of \LSTPGSa{} when it behaves as an algorithm for general \SMCC{}, as stated in \cref{theorem:LSTPGS general results}. This occurs when one or more of its input parameters are too large (or when $\Error$ is too small) for \LSTPGSa{} to gain any performance advantage from the $\PSep$-superseparability of $\Funcf$; the theorem gives precise conditions in terms of the parameters for this to occur. In this case, \LSTPGSa{} behaves the same way as \LSPGBa{} \cite{chen2021best} except for its use of the subroutine \TBS{} instead of \TS{}.
	
	\begin{theorem} \label{theorem:LSTPGS general results}
		If $\GroundSetSize \leq \lceil \frac{\PSep \Card}{1-\ApproxParamA} + \Card \rceil$ and $\GroundSetSize \leq \frac{\PSep \Card}{(1-\Error) \Error / 2} $, then \LSTPGSa{} behaves as an algorithm for general \SMCC{} and, with probability $1 - \frac{2}{\GroundSetSize}$, achieves:
		\begin{itemize}
			\item a solution $\SetS$ satisfying $|\SetS| \leq \Card$ and $\Funcf(\SetS) \geq ( 1-\frac{1}{e}-\Error ) \OptValue$,
			\item an adaptive complexity of $O( \frac{1}{\Error^2} \log(\frac{\GroundSetSize}{\Error}) )$, and
			\item an expected query complexity of $O( \frac{\GroundSetSize}{\Error} + \frac{\log \GroundSetSize}{\Error^2} )$.
		\end{itemize}
	\end{theorem}
	
	\begin{proof}
		If we assume the condition $\GroundSetSize \leq \lceil \frac{\PSep \Card}{1-\ApproxParamA} + \Card \rceil$ holds, then \LSTPGSa{} will run \LS{} on all of $\GroundSet$ as a pre-processing step.
		
		Now assume the condition $\GroundSetSize \leq \frac{\PSep \Card}{(1-\Error) \Error / 2} $ also holds. When \LSTPGSa{} runs \PGSa{}, \cref{line:PGS thresholding} of \PGSa{} will always be executed since, for every iteration $\IterIndex$, $|\GoodElems_{\IterIndex}| \leq \GroundSetSize \leq \frac{\PSep \Card}{(1-\Error) \Error / 2} < \MinSizeForSample_{\IterIndex}$. This execution closely follows that of \PGBa{} with the only significant difference being that \cref{line:PGS thresholding} of \PGSa{} calls \TBSa{} instead of \TSa{}. This is only relevant to the query complexity.
		
		Since the query complexity of \TBSa{} is $O( \GroundSetSize + \frac{\log \GroundSetSize}{\Error} )$ and there are $O(\log_{1-\Error}(\ApproxParam) ) = O( \frac{1}{\Error} )$ iterations of the \cref{line:PGS while} loop in \PGSa{} (\cref{claim:num PGS iters}), the overall expected query complexity of \PGSa{} is $O( \frac{\GroundSetSize}{\Error} + \frac{\log \GroundSetSize}{\Error^2} )$. The remaining guarantees of \PGSa{} are the same as \PGBa{}. Combining these with the guarantees of running \LS{} on $\GroundSet$ proves \cref{theorem:LSTPGS general results}. 
	\end{proof}
	
	\section{Parallel Thresholding Procedure for \SMCC{}}\label{sec:thresholdblockseq}
	In this section, we propose \TBS{} (\TBSa), with pseudocode given in \cref{alg:TBS}. This is the subroutine used by \PGS{} in \cref{line:PGS thresholding}. We formally state the performance guarantees of \TBSa{} in \cref{theorem:TBS results} below and prove them in \cref{sec:TBS analysis}, where we also state related Chernoff Bounds and probability lemmas.
	
	\begin{theorem} \label{theorem:TBS results}
		Suppose \TBSa{} (\cref{alg:TBS}) is run such that $\Funcg$ is monotone submodular, $\GroundSubsetSize = |\GoodElems| \leq \MinSizeForSample$, $0 < \ErrorA < 1$, and $0 < \ProbFail < 1$. Then, with probability $1-\frac{\ProbFail}{\MinSizeForSample}$, \TBSa{} achieves:
		\begin{itemize}
			\item an adaptive complexity of $O ( \frac{1}{\ErrorA} \log( \frac{\MinSizeForSample}{\ProbFail} ) ) $,
			\item an expected query complexity of $O( \GroundSubsetSize + \frac{\log \GroundSubsetSize}{\ErrorA} )$,
			\item an output set $\SetT$ satisfying $|\SetT| \leq \CardA$ and $\Funcg(\SetT \mid \varnothing) \geq \frac{1-\ErrorA}{1+\ErrorA} \Threshold |\SetT|$, and
			\item in case $|\SetT| < \CardA$, for all $ \Elem \in \GoodElems \colon \Funcg(\Elem \mid \SetT) < \Threshold $.
		\end{itemize}
	\end{theorem}
	
	Given a value oracle $ \Funcg $, an input set $\GoodElems$, a value $\MinSizeForSample \geq |\GoodElems|$, an error term $ \ErrorA $, and a probability term $ \ProbFail $, the purpose of \TBSa{} is to return a set $ \SetT \subseteq \GoodElems $ satisfying $ \Funcg(\SetT \mid \varnothing) \geq \frac{1-\ErrorA}{1+\ErrorA}\Threshold |\SetT| $ in $ O ( \frac{1}{\ErrorA} \log(\frac{\MinSizeForSample}{\ProbFail}) ) $ adaptive rounds. The task of finding a set whose average marginal gain is above some threshold is common in many algorithms for submodular maximisation; in fact, \textsc{\TBSa} is an improved version of the \TS{} (\TSa{}) procedure by Chen et al. \cite{chen2021best} for performing this task and can serve to replace it. The main feature of \TBSa{} is its query complexity of $O( \GroundSubsetSize + \frac{\log \GroundSubsetSize}{\ErrorA} )$, which has an improved dependence on $\ErrorA^{-1}$ over that of \TSa{}.
	
	For the purpose of comparison, we briefly describe \TSa{} and point out its main inefficiency, which leads to its $ O( \frac{\GroundSubsetSize}{\ErrorA} ) $ query complexity. After this, we describe the steps in \TSa{}, and finally explain how these steps work to achieve its improved query complexity of $O( \GroundSubsetSize + \frac{\log \GroundSubsetSize}{\ErrorA} )$.
	
	\paragraph{Description of \TS{}.} \TSa{} works by updating a solution $\SetT$ over a loop. Each loop iteration uses an improved \emph{adaptive sequencing} technique to update $\SetT$. Specifically, each loop iteration (1) queries $ \Funcg $ over all previously remaining elements $\Elem$ to filter out those $\Elem$ with $\Funcg(\Elem \mid \SetT) < \Threshold$, (2) uniformly-at-random permutes the remaining elements, and then (3) adds an appropriate prefix of the remaining elements to $\SetT$. By adding this prefix, at least $ \frac{\ErrorA}{2} $ proportion of the remaining elements now have $\Funcg(\Elem \mid \SetT) < \Threshold$ (with probability $\geq \frac{1}{2}$) and are, thus, filtered out in the next iteration.
	
	\paragraph{Query Complexity of \TS{}.} The main inefficiency in \TSa{} is due to performing filtering queries over \emph{all} remaining elements when only~$\frac{\ErrorA}{2}$ of these elements are likely to be filtered out; in other words, $1-\frac{\ErrorA}{2}$ proportion of elements that \TSa{} queries will \emph{not} be filtered out and, thus, will appear in the next iteration to be queried again. So over all iterations, the expected query complexity due to filtering steps is essentially $O\left( \sum_{\SampleIndex=0}^{\infty} \GroundSubsetSize \left(1-\frac{\ErrorA}{2}\right)^\SampleIndex \right) = O\left(\frac{\GroundSubsetSize}{\ErrorA}\right)$.
	
	\paragraph{Description of \TBS{}.}
	\TBSa{} works by updating a solution $\SetT_{\OuterIterIndex, \InnerIterIndex}$ over an \emph{outer} loop (\cref{line:TBS outer iterations}) that contains a nested \emph{inner} loop (\cref{line:TBS inner iterations}). Each outer iteration $\OuterIterIndex$ updates the set of remaining elements $\RemElemsDefer_{\OuterIterIndex}$ by filtering out those $\Elem \in \RemElems_{\OuterIterIndex-1}$ with $\Funcg(\Elem \mid \SetT_{\OuterIterIndex-1, \InnerIterIndex}) < \Threshold$ (\cref{line:TBS filtering step}). Each inner iteration $\InnerIterIndex$ uniformly-at-random samples a ``block'' $\Block$ of size $O( \lceil \ErrorA |\RemElems_{\IterIndex}| \rceil )$ from $\RemElems_{\OuterIterIndex}$ (\cref{line:TBS sample block}), and filters out those $\Elem\in\Block$ with $\Funcg(\Elem \mid \SetT_{\OuterIterIndex, \InnerIterIndex-1}) < \Threshold$ to give $\BlockFiltered$ (\cref{line:TBS block filtered}); that is, $\BlockFiltered$ is obtained by rejection sampling. Then the inner iteration adds an appropriate prefix $\Prefix_{\PrefixIndexBest} \subseteq \BlockFiltered$ to $\SetT_{\OuterIterIndex, \InnerIterIndex-1}$, giving $\SetT_{\OuterIterIndex, \InnerIterIndex}$ (\crefrange{line:TBS filtered block permutation}{line:TBS T update}).
	
	\paragraph{Achieving the Query Complexity of \TBS{}.}
	At a high level, \TBSa{} uses the same adaptive sequencing technique as \TSa{}, but improves the query complexity's dependence on $ \ErrorA^{-1} $ essentially because each \emph{outer} iteration (which performs ``filtering'' queries over all remaining $\Elem$) is only executed when a \emph{constant} $\ReduProp$ proportion of $\Elem$ are likely to be filtered out, i.e., satisfy $\Funcg(\Elem \mid \SetT_{\OuterIterIndex-1, \InnerIterIndex}) < \Threshold$. The fact that $\ReduProp$ proportion of $\Elem$ are likely to satisfy this is achieved by the \emph{inner} loop.
	
	Below, we give a simplified explanation of why \TBSa{} has an expected query complexity of only $ O( \GroundSubsetSize + \frac{\log \GroundSubsetSize}{\ErrorA} ) $, with details in \cref{sec:TBS analysis}. Note that \TBSa{} performs $\NumIters \in O( \log(\frac{\MinSizeForSample}{\ProbFail}) )$ outer iterations and $\NumBlocks \in O( \frac{1}{\ErrorA} )$ inner iterations (\cref{line:TBS initialisation}), the latter being important to our explanation.
	\begin{itemize}
		\item In each \emph{inner} iteration $\InnerIterIndex$, adding the prefix $\Prefix_{\PrefixIndexBest}$ causes $ \geq \frac{\ErrorA}{4} $ proportion of $ \Elem \in \RemElems_{\OuterIterIndex} $ with $ \Funcg(\Elem \mid \SetT_{\OuterIterIndex, \InnerIterIndex-1}) \geq \Threshold $ to have $ \Funcg(\Elem \mid \SetT_{\OuterIterIndex, \InnerIterIndex}) < \Threshold $ (with probability $ \geq \frac{1}{2} $). Further, each \emph{inner} iteration uses $O(|\Block|) = O( \lceil \ErrorA |\RemElems_{\OuterIterIndex}| \rceil ) $ queries since queries are only made on $ \Elem \in \Block $ or on prefixes of $\Block^*$, the filtered subset of $\Block$.
		\item In each \emph{outer} iteration $\IterIndex$, with probability $ \geq \frac{1}{2} $, $O( \frac{1}{\ErrorA} )$ inner iterations will successfully cause $ \geq \frac{\ErrorA}{4} $ proportion of $ \Elem \in \RemElems_{\OuterIterIndex} $ with $ \Funcg(\Elem \mid \SetT_{\OuterIterIndex, \InnerIterIndex-1}) \geq \Threshold $ to have $ \Funcg(\Elem \mid \SetT_{\OuterIterIndex, \InnerIterIndex}) < \Threshold $. Thus, by the start of the next \emph{outer} iteration $\IterIndex+1$, at least $ 1 - ( 1- \frac{\ErrorA}{4} )^{O(\frac{1}{\ErrorA} )} $ proportion of $ \Elem \in \RemElems_{\OuterIterIndex} $ has $ \Funcg(\Elem \mid \SetT_{\OuterIterIndex, \InnerIterIndex}) < \Threshold $ (with probability $ \geq \frac{1}{2} $). This proportion is at least a constant $ \ReduProp $ for all $ \ErrorA > 0 $, so this next \emph{outer} iteration will filter out a constant $\ReduProp$ proportion of $\Elem \in \RemElems_{\OuterIterIndex}$.
		\item Furthermore, over a single \emph{outer} iteration $\IterIndex$, only $O(\NumBlocks \, |\Block| )$ queries are made, which evaluates to $ O(\frac{1}{\ErrorA} \lceil \ErrorA |\RemElems_{\OuterIterIndex}| \rceil ) = O( |\RemElems_{\OuterIterIndex}| + \frac{1}{\ErrorA} ) $ queries.
		
		\item Initially $ |\RemElems_0| = \GroundSubsetSize $. So due to the reduction of constant $\ReduProp$ proportion of $|\RemElems_{\OuterIterIndex}|$ in each outer iteration, the overall expected query complexity of \textsc{\TBSa} is essentially $O\left( \sum_{\SampleIndex=0}^{\infty} \GroundSubsetSize \left(1-\ReduProp\right)^\SampleIndex + \frac{1}{\ErrorA} \log_{\frac{1}{1-\ReduProp}}(\GroundSubsetSize) \right) = O\left( \GroundSubsetSize + \frac{\log \GroundSubsetSize}{\ErrorA} \right)$.
	\end{itemize}
	
	\begin{algorithm}
		\caption{} \label{alg:TBS}
		\begin{algorithmic}[1]
			\Procedure{\TBS}{$\Funcg, \GoodElems, \MinSizeForSample, \CardA, \ErrorA, \ProbFail, \Threshold$}
			\Input{value oracle $ \Funcg \colon 2^{\GoodElems} \rightarrow \mathbb{R}_{\geq 0} $, set of elements $ \GoodElems $, quantity $ \MinSizeForSample $ such that $ |\GoodElems| \leq \MinSizeForSample $, cardinality constraint $ \CardA $, error term $ \ErrorA $, failure probability term $ \ProbFail $, marginal gain threshold $ \Threshold $}
			\Output{set $ \SetT \subseteq \GoodElems $ satisfying $ \Funcg(\SetT \mid \varnothing) \geq \frac{1-\ErrorA}{1+\ErrorA}\Threshold |\SetT| $}
			\StateInd{$ \OuterIterIndex \gets 0 $, $ \InnerIterIndex \gets 0 $, $ \SetT_{0, 0} \gets \varnothing $, $ \RemElems_{0} \gets \GoodElems $, $ \ReduProp \gets \ReduPropVal $, $ \NumIters \gets \left\lceil 4 \left( 1 + \frac{1}{\ReduProp} \right) \log \left( \frac{\MinSizeForSample}{\ProbFail} \right) \right\rceil $, $\NumBlocks \gets \left\lceil 4 \left( 1 + \frac{4}{\ErrorA} \right) \log\left(\frac{2}{1-\ReduProp}\right) \right\rceil $} \label{line:TBS initialisation}
			
			\For{$ \NumIters $ \textnormal{iterations}}\label{line:TBS outer iterations}
			\State $ \OuterIterIndex \gets \OuterIterIndex + 1 $ \label{line:TBS outer iter index update}
			\State $ \RemElems_{\OuterIterIndex} \gets \{ \Elem \in \RemElems_{\OuterIterIndex-1} : \Funcg( \Elem \mid \SetT_{\OuterIterIndex-1, \InnerIterIndex}) \geq \tau \} $ \label{line:TBS filtering step}
			\If{$ |\RemElems_{\OuterIterIndex}| = 0 $}\label{line:TBS check rem elems}
			\State \Return{$ \SetT_{\OuterIterIndex-1,\InnerIterIndex} $}
			\EndIf
			\State $ \SetT_{\OuterIterIndex,0} \gets \SetT_{\OuterIterIndex-1, \InnerIterIndex} $
			\State $ \InnerIterIndex \gets 0 $
			\For{$\NumBlocks$ \textnormal{iterations}}\label{line:TBS inner iterations}
			\State $ \InnerIterIndex \gets \InnerIterIndex + 1 $ \label{line:TBS inner iter index update}
			\State $ \Block \gets \text{uniform-at-random sample of } {\left\lceil \frac{\ErrorA}{4(1-4\ReduProp)} |\RemElems_{\OuterIterIndex}| \right\rceil} \text{ elements from } \RemElems_{\OuterIterIndex} $ \label{line:TBS sample block}
			\State $ \BlockFiltered \gets \{ \Elem \in \Block : \Funcg( \Elem \mid \SetT_{\OuterIterIndex,\InnerIterIndex-1}) \geq \Threshold \} $ \label{line:TBS block filtered}
			\If{$|\BlockFiltered| = 0$}
			\State \textbf{continue to next iteration}
			\EndIf
			\State $ \{ \BlockElem_1, \dots, \BlockElem_{ |\BlockFiltered| } \} \gets \text{uniform-at-random permutation of } \BlockFiltered $ \label{line:TBS filtered block permutation}
			\State $ \MaxPrefixLength \gets \min\{ \CardA-|\SetT_{\OuterIterIndex,\InnerIterIndex-1}|, |\BlockFiltered| \} $ \label{line:TBS maxindex}
			\State $ \PrefixIndices \gets \{ \lfloor (1+\ErrorA)^\GeoExp \rfloor : 1 \leq \lfloor (1+\ErrorA)^\GeoExp \rfloor \leq \MaxPrefixLength, \GeoExp \in \mathbb{N} \} \cup \{ \MaxPrefixLength \} $ \label{line:TBS prefixes}
			
			\For{$ \PrefixIndex $ \textnormal{in} $ \PrefixIndices $}
			\State $ \Prefix_\PrefixIndex \gets \{\BlockElem_1, \dots, \BlockElem_\PrefixIndex \} $ \label{line:TBS define prefixes}
			\EndFor
			\State $ \GoodPrefixIndices \gets {\left\{ \PrefixIndex \in \PrefixIndices : \Funcg(\Prefix_\PrefixIndex \, \mid \, \SetT_{\OuterIterIndex,\InnerIterIndex-1}) \geq (1-\ErrorA) \Threshold |\Prefix_\PrefixIndex| \right\}} $ \label{line:TBS good prefixes}
			\If{$ \max \GoodPrefixIndices < \MaxPrefixLength $} \label{line:TBS check max index}
			\State $ \PrefixIndexBest \gets \min\{ \PrefixIndex \in \PrefixIndices : \forall \PrefixIndexGood \in \GoodPrefixIndices, \PrefixIndex > \PrefixIndexGood \} $ \label{line:TBS assign ibest succ}
			\Else
			\State $ \PrefixIndexBest \gets \MaxPrefixLength $ \label{line:TBS assign ibest max}
			\EndIf
			\State $ \SetT_{\OuterIterIndex,\InnerIterIndex} \gets \SetT_{\OuterIterIndex,\InnerIterIndex-1} \cup \Prefix_{\PrefixIndexBest} $ \label{line:TBS T update}
			\If{$ |\SetT_{\OuterIterIndex,\InnerIterIndex}| = \CardA $}
			\State \Return{$ \SetT_{\OuterIterIndex,\InnerIterIndex} $} \label{line:TBS return set tight card}
			\EndIf
			\EndFor
			\EndFor
			\State \Return{failure} \label{line:TBS return failure}
			\EndProcedure
		\end{algorithmic}
	\end{algorithm}
	
	\subsection{Analysis of \TBS{}} \label{sec:TBS analysis}
	\paragraph{Chernoff Bounds and Probability Lemmas for \TBSa{}.}
	In this section, we state Chernoff bounds in \cref{lemma:chernoff bounds}. We also state \cref{lemma:1st lemma for bernoulli trials,lemma:2nd lemma for bernoulli trials} for replacing dependent Bernoulli trials with independent Bernoulli trials; the proof of these two lemmas are given by Chen et al. \cite{chen2021best}. The lemmas are used in analyzing \TBSa{}, particularly its query complexity and its success probability.
	
	\begin{lemma}[Chernoff bounds \cite{Mitzenmacher2017}] \label{lemma:chernoff bounds}
		Let $ \IndepTrial_1, \dots, \IndepTrial_\NumRandVar $ be independent binary random variables such that $ \Pr[\IndepTrial_{\RandVarIndex} = 1] = \RandVarProb_{\RandVarIndex} $. Let $ \Mean = \sum_{\RandVarIndex}^{\NumRandVar} \RandVarProb_{\RandVarIndex} $. Then for every $ \Dev \geq 0 $:
		\begin{align}
			\Pr\left[ \sum_{\RandVarIndex=1}^{\NumRandVar} \IndepTrial_{\RandVarIndex} \geq (1+\Dev)\Mean \right] \leq e^{-\frac{\Dev^2 \Mean}{2+\Dev}}\,. \label{eqn:chernoff bound greater}
		\end{align}
		Moreover, for every $ 0 \leq \Dev \leq 1 $:
		\begin{align}
			\Pr\left[ \sum_{\RandVarIndex=1}^{\NumRandVar} \IndepTrial_{\RandVarIndex} \leq (1-\Dev)\Mean \right] \leq e^{-\frac{\Dev^2 \Mean}{2}}\,. \label{eqn:chernoff bound less}
		\end{align}
	\end{lemma}
	
	\begin{lemma}[1st Lemma for Bernoulli trials \cite{chen2021best}] \label{lemma:1st lemma for bernoulli trials}
		Let $ \DepTrial_{1}, \dots, \DepTrial_{\NumRandVar} $ be a sequence of $ \NumRandVar $ Bernoulli trials where the probability of $ \DepTrial_{\RandVarIndex} = 1 $ depends on the results of the previous trials $ \DepTrial_{1}, \dots, \DepTrial_{\RandVarIndex-1} $. Suppose that for some constant $ 0 < \ProbTrialSucc \leq 1 $ and every $ \depTrial_{1},\dots,\depTrial_{\RandVarIndex-1} \in \{0, 1\} $, we have 
		\begin{align}
			\Pr[\DepTrial_{\RandVarIndex}=1 \mid \DepTrial_{1}=\depTrial_{1}, \dots, \DepTrial_{\RandVarIndex-1}=\depTrial_{\RandVarIndex-1}] \geq \ProbTrialSucc.
		\end{align}
		Now let $ \IndepTrial_{1}, \dots, \IndepTrial_{\NumRandVar} $ be a sequence of independent Bernoulli trials such that for all $ 1 \leq \RandVarIndex \leq \NumRandVar : \Pr[\IndepTrial_{\RandVarIndex}] = \ProbTrialSucc $. Then for an arbitrary integer $ \NumSucc $,
		\begin{align}
			\Pr\left[ \sum_{\RandVarIndex=1}^{\NumRandVar} \DepTrial_{\RandVarIndex} \leq \NumSucc \right] \leq \Pr\left[ \sum_{\RandVarIndex=1}^{\NumRandVar} \IndepTrial_{\RandVarIndex} \leq \NumSucc \right]\,.
		\end{align}
		Moreover, let $ \FirstSucc $ be the first index $ \RandVarIndex $ such that $ \DepTrial_\RandVarIndex = 1 $. Then,
		\begin{align}
			\mathbb{E}[\FirstSucc] \leq \frac{1}{\ProbTrialSucc}\,. \label{eqn:expected number of trials}
		\end{align}
	\end{lemma}
	
	\begin{lemma}[2nd Lemma for Bernoulli trials \cite{chen2021best}] \label{lemma:2nd lemma for bernoulli trials}
		Let $ \DepTrial_{1}, \dots, \DepTrial_{\NumRandVar+1} $ be a sequence of $ \NumRandVar+1 $ Bernoulli trials where the probability of $ \DepTrial_{\RandVarIndex} = 1 $ depends on the results of the previous trials $ \DepTrial_{1}, \dots, \DepTrial_{\RandVarIndex-1} $, and it decreases from $ 1 $ to $ 0 $. Let $ \RandVarMaxIndex $ be a random variable dependent on the outcomes of the $ \NumRandVar+1 $ Bernoulli trials. Suppose that for some constant $ 0 < \ProbTrialSucc < 1 $ and every $ \depTrial_{1},\dots,\depTrial_{\RandVarIndex-1} \in \{0, 1\} $, we have
		\begin{align}
			\Pr[\DepTrial_{\RandVarIndex}=1 \mid \DepTrial_{1}=\depTrial_{1}, \dots, \DepTrial_{\RandVarIndex-1}=\depTrial_{\RandVarIndex-1}, \RandVarIndex \leq \RandVarMaxIndex] \geq \ProbTrialSucc\,.
		\end{align}
		Now let $ \IndepTrial_{1}, \dots, \IndepTrial_{\NumRandVar} $ be a sequence of independent Bernoulli trials such that for all $ 1 \leq \RandVarIndex \leq \NumRandVar : \Pr[\IndepTrial_{\RandVarIndex}] = \ProbTrialSucc $. Then for an arbitrary proportion $ 0 \leq \NumSucc \leq 1 $,
		\begin{align}
			\Pr\left[ \sum_{\RandVarIndex=1}^{\RandVarMaxIndex} \DepTrial_{\RandVarIndex} \leq \NumSucc \RandVarMaxIndex \right] \leq \Pr\left[ \sum_{\RandVarIndex=1}^{\RandVarMaxIndex} \IndepTrial_{\RandVarIndex} \leq \NumSucc \RandVarMaxIndex \right]\,.
		\end{align}
	\end{lemma}
	
	\paragraph{Adaptive Complexity of \TBSa{}.}
	\begin{lemma} \label{lemma:TBS adaptive complexity}
		The adaptive complexity of \TBSa{} (\cref{alg:TBS}) is $ O ( \frac{1}{\ErrorA}\log ( \frac{\MinSizeForSample}{\ProbFail} ) ) $.
	\end{lemma}
	
	\begin{proof}
		We call an iteration of the \cref{line:TBS outer iterations} loop an \emph{outer iteration} and an iteration of the \cref{line:TBS inner iterations} loop an \emph{inner iteration}.
		
		Observe that each inner iteration requires 2 adaptive rounds: \cref{line:TBS block filtered} uses 1 round of $ |\Block| + 1 $ queries, and \cref{line:TBS good prefixes} uses 1 round of at most $ |\Block| $ queries. Then observe that each outer iteration requires $ 1 + 2\cdot\NumBlocks $ adaptive rounds: \cref{line:TBS filtering step} uses 1 round of $ |\RemElems_{\OuterIterIndex}| + 1 $ queries, and there are $ \NumBlocks $ sequential inner iterations each requiring 2 adaptive rounds. Finally, \TBSa{} overall requires $ \NumIters \cdot (1+2 \cdot \NumBlocks) $ adaptive rounds: there are $ \NumIters $ sequential outer iterations, each requiring $ (1+2 \cdot \NumBlocks) $ adaptive rounds. Thus, we upper-bound the adaptive complexity of \TBSa{}.
		\begin{align}
			&\text{Adaptive complexity of \textsc{\TBSa}} \notag \\
			&= \NumIters \cdot (1+2 \cdot \NumBlocks) \notag \\
			&=  \left\lceil 4 \left( 1 + \frac{1}{\ReduProp} \right) \log \left( \frac{\MinSizeForSample}{\ProbFail} \right) \right\rceil \cdot \left( 1 + 2 \cdot \left\lceil 4 \left( 1 + \frac{4}{\ErrorA} \right) \log\left(\frac{2}{1-\ReduProp}\right) \right\rceil \right) &\text{\cref{line:TBS initialisation}} \notag \\
			&= O \left( \frac{1}{\ErrorA}\log \left( \frac{\MinSizeForSample}{\ProbFail} \right) \right)\,. &\text{$\ReduProp$ is constant} \notag
		\end{align} 
	\end{proof}
	
	\paragraph{Query Complexity of \TBSa{}.}
	\begin{lemma} \label{lemma:TBS query complexity}
		Suppose \TBSa{} (\cref{alg:TBS}) is run such that $0 < \ErrorA < 1$, and let $\GroundSubsetSize=|\GoodElems|$. Then the expected query complexity of \TBSa{} is $ O(\GroundSubsetSize + \frac{\log \GroundSubsetSize}{\ErrorA}) $.
	\end{lemma}
	
	\begin{proof}
		We call an iteration of the \cref{line:TBS outer iterations} loop an \emph{outer iteration} and an iteration of the \cref{line:TBS inner iterations} loop an \emph{inner iteration}. We first prove \cref{claim:TBS query complexity outer iter} below.
		
		\begin{claim}\label{claim:TBS query complexity outer iter}
			The query complexity of the $\OuterIterIndex$th outer iteration of \TBSa{} is at most $( 1 + \frac{ 10 \log(\frac{2}{1-\ReduProp}) + \frac{1}{2} }{1-4\ReduProp} ) |\RemElems_{\OuterIterIndex-1}| + \frac{60}{\ErrorA} \log(\frac{2}{1-\ReduProp}) + 4$.
		\end{claim}
		
		\begin{proof}
			In the $ \OuterIterIndex $th outer iteration, the \cref{line:TBS filtering step} filtering step performs $ |\RemElems_{\OuterIterIndex-1}| $ queries plus $ 1 $ query to compute $ \Funcf(\SetT_{\OuterIterIndex-1, \InnerIterIndex}) $ for all marginal gains. Then observe that each inner iteration performs $ |\Block| = \lceil \frac{\ErrorA}{4(1-4\ReduProp)} |\RemElems_{\OuterIterIndex}| \rceil $ queries in \cref{line:TBS block filtered} and at most $ |\Block| = \lceil \frac{\ErrorA}{4(1-4\ReduProp)} |\RemElems_\OuterIterIndex| \rceil $ queries in \cref{line:TBS good prefixes}, plus $ 1 $ query to compute $ \Funcf(\SetT_{\OuterIterIndex, \InnerIterIndex-1}) $ for all marginal gains in \cref{line:TBS block filtered,line:TBS good prefixes}. Further, for each outer iteration, there are at most $ \NumBlocks = \lceil 4 ( 1 + \frac{4}{\ErrorA} ) \log(\frac{2}{1-\ReduProp}) \rceil $ inner iterations.
			
			Thus, we bound the overall query complexity of the $ \OuterIterIndex $th outer iteration below. \cref{eqn:TBS query complexity outer iter} follows since $ \ErrorA < 1 $, and $|\RemElems_\OuterIterIndex| \leq |\RemElems_{\OuterIterIndex-1}| $ holds by submodularity. It is worth noting that, in \cref{eqn:TBS divide out error}, we are able cancel out $\frac{1}{\ErrorA}$ by our assignments of $|\Block|$ and $\NumBlocks$ so that it does not appear as a factor of $|\RemElems_{\OuterIterIndex-1}|$.
			\begin{align}
				&\text{Query complexity of $ \OuterIterIndex $th outer iteration} \notag \\
				&\leq |\RemElems_{\OuterIterIndex-1}| + 1 + \NumBlocks \cdot \left( 2|\Block| + 1 \right) \notag \\
				&\leq |\RemElems_{\OuterIterIndex-1}| + 1 + \left\lceil 4 \left( 1 + \frac{4}{\ErrorA} \right) \log\left(\frac{2}{1-\ReduProp}\right) \right\rceil \cdot \left( 2\left\lceil \frac{\ErrorA}{4(1-4\ReduProp)} |\RemElems_\OuterIterIndex| \right\rceil + 1 \right) \notag \\
				&\leq |\RemElems_{\OuterIterIndex-1}| + 1 + \left( 4 \left(\frac{5}{\ErrorA}\right) \log\left(\frac{2}{1-\ReduProp}\right) + 1 \right) \cdot \left( 2\left( \frac{\ErrorA}{4(1-4\ReduProp)} |\RemElems_\OuterIterIndex| + 1 \right) + 1 \right) \notag \\
				&= |\RemElems_{\OuterIterIndex-1}| + 1 + \left( \frac{20}{\ErrorA} \log\left(\frac{2}{1-\ReduProp}\right) + 1 \right) \cdot \left( \frac{\ErrorA}{2(1-4\ReduProp)} |\RemElems_\OuterIterIndex| + 3 \right) \notag \\
				&= |\RemElems_{\OuterIterIndex-1}| + 1 + \frac{\ErrorA\cdot \left( \frac{20}{\ErrorA} \log\left(\frac{2}{1-\ReduProp}\right) + 1 \right) }{2(1-4\ReduProp)} |\RemElems_\OuterIterIndex| + \frac{60}{\ErrorA} \log\left(\frac{2}{1-\ReduProp}\right) + 3 \label{eqn:TBS divide out error} \\
				&= |\RemElems_{\OuterIterIndex-1}| + \frac{ 10 \log\left(\frac{2}{1-\ReduProp}\right) + \frac{\ErrorA}{2} }{1-4\ReduProp} |\RemElems_\OuterIterIndex| + \frac{60}{\ErrorA} \log\left(\frac{2}{1-\ReduProp}\right) + 4 \notag \\
				&\leq \left( 1 + \frac{ 10 \log\left(\frac{2}{1-\ReduProp}\right) + \frac{1}{2} }{1-4\ReduProp} \right) |\RemElems_{\OuterIterIndex-1}| + \frac{60}{\ErrorA} \log\left(\frac{2}{1-\ReduProp}\right) + 4\,. \label{eqn:TBS query complexity outer iter}
			\end{align} 
		\end{proof}
		
		Before we bound the overall expected query complexity of \TBSa{}, we define the random variable $ \NumRemElemsInter_\PropExp $, for every $ \PropExp \geq 0 $, as the number of outer iterations $ \OuterIterIndex $ for which $ (1-\ReduProp)^{\PropExp} |\GoodElems| \geq |\RemElems_{\OuterIterIndex}| > (1-\ReduProp)^{\PropExp+1} |\GoodElems| $. Importantly, we have the following claim.
		
		\begin{claim} \label{claim:num outer iters ub}
			For every $\PropExp \geq 1$, we have that $ \mathbb{E}[\NumRemElemsInter_\PropExp] \leq 2 $. 
		\end{claim}
		
		\begin{proof}
			An $\OuterIterIndex$th outer iteration successfully results in $ |\RemElems_{\OuterIterIndex+1}| \leq (1-\ReduProp) |\RemElems_{\OuterIterIndex}| $ with probability at least $ \frac{1}{2} $ by \cref{lemma:TBS outer failure}. Then, the expected number $\FirstSucc$ of outer iterations until the next successful iteration is at most $ 2 $ by \cref{lemma:1st lemma for bernoulli trials}, \cref{eqn:expected number of trials}. Therefore, in expectation, there are at most $2$ outer iterations for which $|\RemElems_{\OuterIterIndex}|$ falls inside $ (1-\ReduProp)^{\PropExp} |\GoodElems| \geq |\RemElems_{\OuterIterIndex}| > (1-\ReduProp)^{\PropExp+1} |\GoodElems| $, meaning $\mathbb{E}[\NumRemElemsInter_\PropExp ] \leq 2$. 
		\end{proof}
		
		Now we upper-bound the overall expected query complexity of \TBSa{}. Below, \cref{eqn:sum query complexity outer iter} follows from \cref{claim:TBS query complexity outer iter}, \cref{eqn:num outer iters ub} follows from \cref{claim:num outer iters ub}, \cref{eqn:geometric series good elems} follows from the geometric series, and \cref{eqn:TBS final query complexity} follows because $\ReduProp$ is constant.
		\begin{align}
			&\mathbb{E}[\text{Query complexity of \textsc{\TBSa}}] \notag \\
			&= \mathbb{E} \left[ \sum_{\OuterIterIndex \colon |\RemElems_{\OuterIterIndex-1}| \geq 1} \left( \text{Query complexity of $ \OuterIterIndex $th outer iteration} \right) \right] \notag \\
			& \leq \mathbb{E} \left[ \sum_{\OuterIterIndex \colon |\RemElems_{\OuterIterIndex-1}| \geq 1} \left( \left( 1 + \frac{ 10 \log\left(\frac{2}{1-\ReduProp}\right) + \frac{1}{2} }{1-4\ReduProp} \right) |\RemElems_{\OuterIterIndex-1}| + \frac{60}{\ErrorA} \log\left(\frac{2}{1-\ReduProp}\right) + 4 \right) \right] \label{eqn:sum query complexity outer iter} \\
			&\leq \sum_{|\GoodElems| (1-\ReduProp)^{\PropExp} \geq 1} \mathbb{E} \left[\NumRemElemsInter_\PropExp\right] \left( \left( 1 + \frac{ 10 \log\left(\frac{2}{1-\ReduProp}\right) + \frac{1}{2} }{1-4\ReduProp} \right) (1-\ReduProp)^{\PropExp} |\GoodElems| + \frac{60}{\ErrorA} \log\left(\frac{2}{1-\ReduProp}\right) + 4 \right) \notag \\
			&\leq \left( 2 + \frac{ 20 \log\left(\frac{2}{1-\ReduProp}\right) + 1 }{1-4\ReduProp} \right) |\GoodElems| \sum_{\PropExp \geq 0}^{\infty} (1-\ReduProp)^{\PropExp} + \sum_{|\GoodElems| (1-\ReduProp)^{\PropExp} \geq 1} \left( \frac{120}{\ErrorA} \log\left(\frac{2}{1-\ReduProp}\right) + 8 \right) \label{eqn:num outer iters ub} \\
			&\leq \left( 2 + \frac{ 20 \log\left(\frac{2}{1-\ReduProp}\right) + 1 }{1-4\ReduProp} \right) \frac{|\GoodElems|}{\ReduProp} + \left(1 + \log_{\frac{1}{1-\ReduProp}}(|\GoodElems|) \right) \left( \frac{120}{\ErrorA} \log\left(\frac{2}{1-\ReduProp}\right) + 8 \right) \label{eqn:geometric series good elems} \\
			&= O\left(\GroundSubsetSize + \frac{\log \GroundSubsetSize}{\ErrorA}\right)\,. \label{eqn:TBS final query complexity}
		\end{align} 
	\end{proof}
	
	\paragraph{Explanation of the Value of $\ReduProp$.} \label{sec:explanation of reduprop}
	\TBSa{} (\cref{alg:TBS}) assigns $\ReduProp$ so as to minimize the factor of $|\GoodElems|=\GroundSubsetSize$ in \cref{eqn:geometric series good elems}, subject to the constraint $0 < \ReduProp < \frac{1}{4}$; this constraint ultimately comes from $|\Block|=\lceil \frac{\ErrorA}{4(1-4\ReduProp)} |\RemElems_{\OuterIterIndex}| \rceil$ in \cref{line:TBS sample block} of \TBSa{}, which is relevant to proving \cref{lemma:TBS sample failure}. The minimizing value of $\ReduProp$ can be verified as $\ReduProp\approx\ReduPropVal$, making the factor $\approx 295.8$.
	
	\paragraph{Success Probability of \TBSa{}.}
	In this section, we prove the success probability of \TBSa{} (\cref{alg:TBS}). Our proof follows the same basic style as that of \TS{}'s success probability \cite{chen2021best}, but is more involved due to the ``block sampling'' and the nested loop structure in \TBSa{}. We state the required result in \cref{lemma:TBS success probability}.
	
	We introduce some special notation that is used throughout this section, in addition to the notation taken from the pseudocode of \TBSa{}. Refer to \cref{sec:thresholdblockseq} for an overview of \TBSa{}.
	\begin{itemize}
		\item the $ \OuterIterIndex $th \emph{outer iteration} is the iteration of the \cref{line:TBS outer iterations} loop indexed by $ \OuterIterIndex $ (\cref{line:TBS outer iter index update}).
		\item the $ \InnerIterIndex $th \emph{inner iteration} is the iteration of the \cref{line:TBS inner iterations} loop indexed by $ \InnerIterIndex $ (\cref{line:TBS inner iter index update}).
		\item $ \RemElemsDefer_{\OuterIterIndex, 0} = \RemElems_{\OuterIterIndex} $.
		\item $ \RemElemsDefer_{\OuterIterIndex, \InnerIterIndex} = \{ \Elem \in \RemElems_{\OuterIterIndex} : \Funcg(\Elem \mid \SetT_{\OuterIterIndex, \InnerIterIndex-1} ) \geq \Threshold \} $, which is the set of elements \emph{implicitly} remaining at the beginning of the $\InnerIterIndex$th inner iteration of the $\OuterIterIndex$th outer iteration.
		\item $ \RemElemsDefer_{\OuterIterIndex, \InnerIterIndex}(\Prefix) = \{ \Elem \in \RemElems_{\OuterIterIndex} : \Funcg(\Elem \mid \SetT_{\OuterIterIndex, \InnerIterIndex-1} \cup \Prefix ) \geq \Threshold \} $, which is the set of elements \emph{implicitly} remaining after adding prefix $\Prefix$ to $\SetT_{\OuterIterIndex, \InnerIterIndex-1}$. Note that $\RemElemsDefer_{\OuterIterIndex, \InnerIterIndex+1} = \RemElemsDefer_{\OuterIterIndex, \InnerIterIndex}(\Prefix_{\PrefixIndexBest})$ since $ \SetT_{\OuterIterIndex,\InnerIterIndex} = \SetT_{\OuterIterIndex,\InnerIterIndex-1} \cup \Prefix_{\PrefixIndexBest} $.
	\end{itemize}
	
	Also, in the $ \OuterIterIndex $th outer iteration and $ \InnerIterIndex $th inner iteration, a prefix index $ \PrefixIndex $ (which gives the prefix $ \Prefix_{\PrefixIndex} = \{ \BlockElem_1, \dots, \BlockElem_{ \PrefixIndex } \} $) is
	\begin{itemize}
		\item \emph{bad} if it satisfies $ \PrefixIndex \notin \GoodPrefixIndices $, i.e., $ \Funcg(\Prefix_\PrefixIndex \, \mid \, \SetT_{\OuterIterIndex,\InnerIterIndex-1}) < (1-\ErrorA)\Threshold |\Prefix_\PrefixIndex| $.
		\item \emph{good} if it satisfies $ \PrefixIndex \in \GoodPrefixIndices $, i.e., $ \Funcg(\Prefix_\PrefixIndex \, \mid \, \SetT_{\OuterIterIndex,\InnerIterIndex-1}) \geq (1-\ErrorA)\Threshold |\Prefix_\PrefixIndex| $.
	\end{itemize}
	
	Finally, letting $\Prefix_{\PrefixElemIndex-1}=\{\BlockElem_1, \dots, \BlockElem_{\PrefixElemIndex-1} \}$, an element $ \BlockElem_{\PrefixElemIndex} \in \Prefix_{\PrefixIndex} = \{\BlockElem_1, \dots, \BlockElem_{\PrefixIndex} \} $ is
	\begin{itemize}
		\item \emph{bad} if it satisfies $ \Funcg( \BlockElem_{\PrefixElemIndex} \mid \SetT_{\OuterIterIndex, \InnerIterIndex-1} \cup \Prefix_{\PrefixElemIndex-1} ) < \Threshold $.
		\item \emph{good} if it satisfies $ \Funcg( \BlockElem_{\PrefixElemIndex} \mid \SetT_{\OuterIterIndex, \InnerIterIndex-1} \cup \Prefix_{\PrefixElemIndex-1} ) \geq \Threshold $.
	\end{itemize}
	
	To aid our analysis, we define three types of `failure' events that can occur in the execution of \textsc{\TBSa}. Also, to simplify our analysis, assume that for every $ \OuterIterIndex $th outer iteration and every $ \InnerIterIndex $th inner iteration we have $ |\SetT_{\OuterIterIndex, \InnerIterIndex}| < \Card $ since otherwise $ \SetT_{\OuterIterIndex, \InnerIterIndex} $ would be returned immediately (\cref{line:TBS return set tight card}).
	
	\begin{description}
		\item[\TBSFailure] This event occurs when \TBSa{} returns \textit{failure}. For this to occur, it is necessary that $ |\RemElems_{\OuterIterIndex}| > 0 $ by the final check of the \cref{line:TBS check rem elems} if-condition $|\RemElems_{\OuterIterIndex}| = 0$.
		
		\item[\OuterFailure] For the $ \OuterIterIndex $th outer iteration, this event occurs when $ |\RemElems_{\OuterIterIndex+1}| >  (1-\ReduProp)|\RemElems_{\OuterIterIndex}| $. This means that the $ \OuterIterIndex $th iteration failed to cause $ \ReduProp $ proportion of the elements in $ \RemElems_{\OuterIterIndex} $ to be filtered out of $ \RemElems_{\OuterIterIndex+1} $.
		
		\item[\InnerFailure] For the $ \OuterIterIndex $th outer iteration and the $ \InnerIterIndex $th inner iteration, this event occurs when one of the following two (disjoint) events occurs:
		\begin{description}
			\item[\SampleFailure] This event occurs when $ |\RemElemsDefer_{\OuterIterIndex, \InnerIterIndex}| > (1-\ReduProp)|\RemElems_{\OuterIterIndex}| $ and $ |\BlockFiltered| < \frac{\ErrorA}{4}|\RemElems_{\OuterIterIndex}| $. This means that an insufficient number of $\Elem \in \RemElemsDefer_{\OuterIterIndex}$ satisfying $\Funcg(\Elem \mid \SetT_{\OuterIterIndex, \InnerIterIndex-1}) \geq \Threshold$ are sampled in $\Block$. To simplify our analysis, this event occurs regardless of whether $ |\RemElemsDefer_{\OuterIterIndex, \InnerIterIndex}(\Prefix_{\PrefixIndexBest})| \leq (1-\frac{\ErrorA}{4})|\RemElemsDefer_{\OuterIterIndex, \InnerIterIndex}| $ or $  |\RemElemsDefer_{\OuterIterIndex, \InnerIterIndex}(\Prefix_{\PrefixIndexBest})| > (1-\frac{\ErrorA}{4})|\RemElemsDefer_{\OuterIterIndex, \InnerIterIndex}| $ occurs. That is, this event ignores whether or not adding the prefix $\Prefix_{\PrefixIndexBest}$ to $\SetT_{\OuterIterIndex, \InnerIterIndex - 1}$ causes $\frac{\ErrorA}{4}$ proportion of $\RemElemsDefer_{\OuterIterIndex, \InnerIterIndex}$ to be implicitly filtered out of $\RemElemsDefer_{\OuterIterIndex, \InnerIterIndex+1}$.
			
			\item[\PrefixFailure] This event occurs when $ |\RemElemsDefer_{\OuterIterIndex, \InnerIterIndex}| > (1-\ReduProp)|\RemElems_{\OuterIterIndex}| $, $ |\BlockFiltered| \geq \frac{\ErrorA}{4}|\RemElems_{\OuterIterIndex}| $ and $  |\RemElemsDefer_{\OuterIterIndex, \InnerIterIndex}(\Prefix_{\PrefixIndexBest})| > (1-\frac{\ErrorA}{4})|\RemElemsDefer_{\OuterIterIndex, \InnerIterIndex}| $. This means that a sufficient number of $\Elem \in \RemElemsDefer_{\OuterIterIndex}$ with $\Funcg(\Elem \mid \SetT_{\OuterIterIndex, \InnerIterIndex-1}) \geq \Threshold$ are sampled in $\Block$; nonetheless adding the prefix $\Prefix_{\PrefixIndexBest}$ to $\SetT_{\OuterIterIndex, \InnerIterIndex - 1}$ fails to cause $\frac{\ErrorA}{4}$ proportion of $\RemElemsDefer_{\OuterIterIndex, \InnerIterIndex}$ to be implicitly filtered out of $\RemElemsDefer_{\OuterIterIndex, \InnerIterIndex+1}$.
		\end{description}
		
		The \sampleFailure{} and \prefixFailure{} events require $ |\RemElemsDefer_{\OuterIterIndex, \InnerIterIndex}| > (1-\ReduProp)|\RemElems_{\OuterIterIndex}| $ since, if some $ \InnerIterIndex $th inner iteration had $ |\RemElemsDefer_{\OuterIterIndex, \InnerIterIndex}| \leq (1-\ReduProp)|\RemElems_{\OuterIterIndex}| $, then it is guaranteed that an \emph{\outerFailure} will \emph{not} occur since, by submodularity, $ \RemElems_{\OuterIterIndex+1} \subseteq \RemElemsDefer_{\OuterIterIndex, \InnerIterIndex} $ and so $ |\RemElems_{\OuterIterIndex+1}| \leq |\RemElemsDefer_{\OuterIterIndex, \InnerIterIndex}| \leq (1-\ReduProp)|\RemElems_{\OuterIterIndex}| $. Thus, it would no longer matter what happens in the remaining inner iterations within the $ \OuterIterIndex $th outer iteration.
	\end{description}
	
	We will prove $\Pr[\text{\SampleFailure}] \leq \frac{1}{4}$ in \cref{lemma:TBS sample failure} and $\Pr[\text{\PrefixFailure}] \leq \frac{1}{4}$ in \cref{lemma:TBS prefix failure}. It will then follow, by a union bound, that $ \Pr[\text{\InnerFailure}] \leq \frac{1}{2} $ in \cref{lemma:TBS inner failure}. We will use this to prove $ \Pr[\text{\OuterFailure}] \leq \frac{1}{2} $ in \cref{lemma:TBS outer failure}. In turn, we will finally use this to prove $ \Pr[\text{\TBSFailure}] \leq \frac{\ProbFail}{\MinSizeForSample} $ in \cref{lemma:TBS success probability}.
	
	\begin{lemma} \label{lemma:TBS sample failure}
		For every $ \OuterIterIndex $th outer iteration and $ \InnerIterIndex $th inner iteration, $ \Pr[\textnormal{\SampleFailure}] \leq \frac{1}{4} $.
	\end{lemma}
	
	\begin{proof}
		First, $ |\RemElemsDefer_{\OuterIterIndex, \InnerIterIndex}| > (1-\ReduProp)|\RemElems_{\OuterIterIndex}| $ is a condition for a \sampleFailure{}. We rewrite this into \cref{eqn:prop bad elements} below. \cref{eqn:prop bad elements} states that the proportion of $ \Elem \in \RemElems_{\OuterIterIndex} $ with $ \Funcg(\Elem \mid \SetT_{\OuterIterIndex, \InnerIterIndex-1}) < \Threshold $ is less than $\ReduProp$.
		\begin{align}
			1-\frac{|\RemElemsDefer_{\OuterIterIndex, \InnerIterIndex}|}{|\RemElems_{\OuterIterIndex}|} &< \ReduProp\,. \label{eqn:prop bad elements}
		\end{align}
		
		$ |\BlockFiltered| < \frac{\ErrorA}{4}|\RemElems_{\OuterIterIndex}| $ is also a condition for a \sampleFailure{}. We rewrite this into \cref{eqn:ratio filtered block} below. \cref{eqn:ratio filtered block} states that the proportion of $\Elem \in \Block$ with $ \Funcg(\Elem \mid \SetT_{\OuterIterIndex, \InnerIterIndex-1}) < \Threshold $ is more than $4\ReduProp$. Note that $4\ReduProp$ is a valid proportion since $\ReduProp=\ReduPropVal$ as assigned in \cref{line:TBS initialisation} of \TBSa{}.
		\begin{align}
			|\BlockFiltered| &< \frac{\ErrorA}{4}|\RemElems_{\OuterIterIndex}|\,, \notag \\
			\frac{|\BlockFiltered|}{|\Block|} &< \frac{\ErrorA}{4} \frac{|\RemElems_{\OuterIterIndex}|}{|\Block|} \notag \\
			&= \frac{\ErrorA |\RemElems_{\OuterIterIndex}|}{4 \left\lceil \frac{\ErrorA}{4(1-4\ReduProp)} |\RemElems_{\OuterIterIndex}| \right\rceil} &\text{size of $\Block$ (\cref{line:TBS sample block})} \notag \\
			&\leq \frac{\ErrorA |\RemElems_{\OuterIterIndex}|}{4 \cdot \frac{\ErrorA}{4(1-4\ReduProp)} |\RemElems_{\OuterIterIndex}|} \notag \\
			&= 1-4\ReduProp\,, \notag \\
			1 - \frac{|\BlockFiltered|}{|\Block|} &> 4\ReduProp\,. \label{eqn:ratio filtered block}
		\end{align}
		
		Since \cref{eqn:prop bad elements} and \cref{eqn:ratio filtered block} are both necessary conditions for a \sampleFailure{}, we can prove $\Pr[\text{\SampleFailure}] \leq \frac{1}{4}$ by assuming \cref{eqn:prop bad elements} holds and then proving
		\begin{align}
			\Pr\left[ 1 - \frac{|\BlockFiltered|}{|\Block|} > 4\ReduProp \right] \leq \frac{1}{4}\,. \notag
		\end{align}
		We prove this using Markov's inequality below, and the fact that $ \mathbb{E}[ 1 - \frac{|\BlockFiltered|}{|\Block|} ] = 1-\frac{|\RemElemsDefer_{\OuterIterIndex, \InnerIterIndex}|}{|\RemElems_{\OuterIterIndex}|} $ since $ \Block $ is a uniform-at-random sample of $ \RemElems_{\OuterIterIndex} $ (\cref{line:TBS sample block}).
		\begin{align}
			\Pr\left[ 1 - \frac{|\BlockFiltered|}{|\Block|} > 4\ReduProp \right] &< \frac{1}{4\ReduProp} \mathbb{E}\left[ 1 - \frac{|\BlockFiltered|}{|\Block|} \right] &\text{Markov's inequality} \notag \\
			&= \frac{1}{4\ReduProp} \left( 1-\frac{|\RemElemsDefer_{\OuterIterIndex, \InnerIterIndex}|}{|\RemElems_{\OuterIterIndex}|} \right) \notag \\
			&< \frac{1}{4\ReduProp} \ReduProp = \frac{1}{4}\,. &\text{\cref{eqn:prop bad elements}} \notag
		\end{align}
		Thus, $\Pr[ 1 - \frac{|\BlockFiltered|}{|\Block|} > 4\ReduProp ] \leq \frac{1}{4}$ and so $\Pr[\text{\SampleFailure}] \leq \frac{1}{4}$. 
	\end{proof}
	
	\begin{lemma} \label{lemma:TBS prefix failure}
		For every $ \OuterIterIndex $th outer iteration and $ \InnerIterIndex $th inner iteration, $ \Pr[\textnormal{\PrefixFailure}] \leq \frac{1}{4} $.
	\end{lemma}
	
	\begin{proof}
		Since $ |\BlockFiltered| \geq \frac{\ErrorA}{4}|\RemElems_{\OuterIterIndex}| $ and $  |\RemElemsDefer_{\OuterIterIndex, \InnerIterIndex}(\Prefix_{\PrefixIndexBest})| > (1-\frac{\ErrorA}{4})|\RemElemsDefer_{\OuterIterIndex, \InnerIterIndex}| $ are both necessary conditions for a \prefixFailure{}, it suffices to prove $ \Pr[\text{\PrefixFailure}] \leq \frac{1}{4}$ by assuming $ |\BlockFiltered| \geq \frac{\ErrorA}{4}|\RemElems_{\OuterIterIndex}| $ and then proving
		\begin{align}
			\Pr\left[ |\RemElemsDefer_{\OuterIterIndex, \InnerIterIndex}(\Prefix_{\PrefixIndexBest})| > \left(1-\frac{\ErrorA}{4}\right)|\RemElemsDefer_{\OuterIterIndex, \InnerIterIndex}| \right] \leq \frac{1}{4}\,. \notag
		\end{align}
		
		We begin by proving the following claim.
		\begin{claim} \label{claim:prob ub num good elements}
			Assuming $ |\BlockFiltered| \geq \frac{\ErrorA}{4}|\RemElems_{\OuterIterIndex}| $, it holds that
			\begin{align}
				&\Pr\left[ |\RemElemsDefer_{\OuterIterIndex, \InnerIterIndex}(\Prefix_{\PrefixIndexBest})| > \left(1-\frac{\ErrorA}{4}\right)|\RemElemsDefer_{\OuterIterIndex, \InnerIterIndex}| \right] \notag \\
				&\leq \Pr\left[\text{$ < (1-\ErrorA) |\Prefix_{\PrefixIndexMax}| $ elements in $ \Prefix_{\PrefixIndexMax} $ are good}\right]\,. \notag
			\end{align}
		\end{claim}
		
		\begin{proof}
			First, let $ \PrefixIndexReq = \min \{ \PrefixIndex \in \mathbb{N} : |\RemElemsDefer_{\OuterIterIndex, \InnerIterIndex}(\Prefix_{\PrefixIndex})| \leq (1-\frac{\ErrorA}{4})|\RemElemsDefer_{\OuterIterIndex, \InnerIterIndex}| \} $. $ \PrefixIndexReq $ must exist since we can always select $\BlockFiltered$ as a prefix (\cref{line:TBS maxindex,line:TBS prefixes}) and we assumed that $ |\BlockFiltered| \geq \frac{\ErrorA}{4}|\RemElems_{\OuterIterIndex}| $. This means that $ |\BlockFiltered| \geq \frac{\ErrorA}{4}|\RemElems_{\OuterIterIndex}| \geq \frac{\ErrorA}{4}|\RemElems_{\OuterIterIndex, \InnerIterIndex}| $, ensuring that $ |\RemElemsDefer_{\OuterIterIndex, \InnerIterIndex}(\BlockFiltered)| \leq (1-\frac{\ErrorA}{4})|\RemElemsDefer_{\OuterIterIndex, \InnerIterIndex}| $.
			
			Observe that $ |\RemElemsDefer_{\OuterIterIndex, \InnerIterIndex}(\Prefix_{\PrefixIndexBest})| > (1-\frac{\ErrorA}{4})|\RemElemsDefer_{\OuterIterIndex, \InnerIterIndex}| $ implies that $ \PrefixIndexBest < \PrefixIndexReq $ since otherwise, for $ \PrefixIndex \geq \PrefixIndexReq $, we would have
			\begin{align}
				|\RemElemsDefer_{\OuterIterIndex, \InnerIterIndex}(\Prefix_{\PrefixIndex})| &\leq |\RemElemsDefer_{\OuterIterIndex, \InnerIterIndex}(\Prefix_{\PrefixIndexReq})| &\text{submodularity and $\Prefix_{\PrefixIndexReq} \subseteq \Prefix_{\PrefixIndex}$} \notag \\
				&\leq \left(1-\frac{\ErrorA}{4}\right)|\RemElemsDefer_{\OuterIterIndex, \InnerIterIndex}|\,. &\text{definition of $\PrefixIndexReq$} \notag
			\end{align}		
				
			Now let $ \PrefixIndexMax = \max\{ \PrefixIndex \in \PrefixIndices : \PrefixIndex < \PrefixIndexReq \} $, where $ \PrefixIndices $ is the set of all available prefix indices (\cref{line:TBS prefixes}). $ \PrefixIndexBest < \PrefixIndexReq $ implies that $ \PrefixIndexMax $ is \emph{bad}, i.e., the prefix $ \Prefix_{\PrefixIndexMax} $ satisfies
			\begin{align}
				\Funcg( \Prefix_{\PrefixIndexMax} \mid \SetT_{\OuterIterIndex, \InnerIterIndex-1} ) < (1-\ErrorA)\Threshold|\Prefix_{\PrefixIndexMax}|\,. \notag
			\end{align}
			This is because if $ \PrefixIndexMax $ was \emph{good}, then either \cref{line:TBS assign ibest succ} would assign the successor of $ \PrefixIndexMax $ to $ \PrefixIndexBest $, or \cref{line:TBS assign ibest max} would assign $ \MaxPrefixLength $ to $ \PrefixIndexBest $ where $\MaxPrefixLength = |\BlockFiltered| \geq \frac{\ErrorA}{4}|\RemElemsDefer_{\OuterIterIndex, \InnerIterIndex}| $, both of which would make $ \PrefixIndexBest \geq \PrefixIndexReq $.
			
			Moreover, $ \PrefixIndexMax $ being bad implies that  $<(1-\ErrorA) |\Prefix_{\PrefixIndexMax}|$ elements in $\Prefix_{\PrefixIndexMax}$ are \emph{good}, as otherwise $\Funcg( \Prefix_{\PrefixIndexMax} \mid \SetT_{\OuterIterIndex, \InnerIterIndex-1} ) < (1-\ErrorA)\Threshold|\Prefix_{\PrefixIndexMax}|$ would not hold.
			
			So from the above discussion, we can conclude that
			\begin{align}
				&\Pr\left[ |\RemElemsDefer_{\OuterIterIndex, \InnerIterIndex}(\Prefix_{\PrefixIndexBest})| > \left(1-\frac{\ErrorA}{4}\right)|\RemElemsDefer_{\OuterIterIndex, \InnerIterIndex}| \right] \notag \\
				&\leq \Pr[\PrefixIndexBest < \PrefixIndexReq] \notag \\
				&\leq \Pr[\text{$ \PrefixIndexMax $ is bad}] \notag \\
				&\leq \Pr\left[\text{$ < (1-\ErrorA) |\Prefix_{\PrefixIndexMax}| $ elements in $ \Prefix_{\PrefixIndexMax} $ are good}\right]\,. \notag
			\end{align} 
		\end{proof}
		
		We will show that $ \Pr[\text{$ < (1-\ErrorA) |\Prefix_{\PrefixIndexMax}| $ elements in $ \Prefix_{\PrefixIndexMax} $ are good}] \leq \frac{1}{4} $ in \cref{claim:prob ub 1/4}. Before this, we make an observation: the process of obtaining each $\BlockElem_{\PrefixElemIndex} \in \Prefix_{\PrefixIndexMax} = \{\BlockElem_1, \dots, \BlockElem_{\PrefixIndexMax} \} $ can be treated as sampling $\BlockElem_{\PrefixElemIndex}$ uniformly at random from $\RemElemsDefer_{\OuterIterIndex, \InnerIterIndex} \setminus \Prefix_{\PrefixElemIndex-1}$. This is because each $\BlockElem_{ \PrefixElemIndex }$ belongs in $\BlockFiltered = \{ \Elem \in \Block : \Funcg( \Elem \mid \SetT_{\OuterIterIndex,\InnerIterIndex-1}) \geq \Threshold \}$, which is a uniform-at-random sample of $\RemElemsDefer_{\OuterIterIndex} \setminus \Prefix_{\PrefixElemIndex-1}$ with rejection (\crefrange{line:TBS sample block}{line:TBS filtered block permutation}). This means $\BlockElem_{\PrefixElemIndex}$ is a uniform-at-random sample from the subset of $\RemElemsDefer_{\OuterIterIndex} \setminus \Prefix_{\PrefixElemIndex-1}$ that satisfies $\Funcg(\Elem \mid \SetT_{\OuterIterIndex, \InnerIterIndex-1}) \geq \Threshold$, which is precisely the set $\RemElemsDefer_{\OuterIterIndex, \InnerIterIndex} \setminus \Prefix_{\PrefixElemIndex-1}$.
		
		\begin{claim} \label{claim:prob ub 1/4}
			$ \Pr[\text{$ < (1-\ErrorA) |\Prefix_{\PrefixIndexMax}| $ elements in $ \Prefix_{\PrefixIndexMax} $ are good}] \leq \frac{1}{4} $.
		\end{claim}
		
		\begin{proof}
			We denote the process of sampling each $ \BlockElem_{\PrefixElemIndex} $ from $ \RemElemsDefer_{\OuterIterIndex, \InnerIterIndex} \setminus \{\BlockElem_1, \dots, \BlockElem_{\PrefixElemIndex-1} \} $ as a Bernoulli trial $ \DepTrial_{\PrefixElemIndex} $ \emph{dependent} on the outcomes of the previous trials, where $ \DepTrial_{\PrefixElemIndex} = 0 $ denotes $ \BlockElem_{\PrefixElemIndex} $ is \emph{bad} and $ \DepTrial_{\PrefixElemIndex} = 1 $ denotes $ \BlockElem_{\PrefixElemIndex} $ is \emph{good}.
			
			Observe that since $ \PrefixElemIndex \leq \PrefixIndexMax < \PrefixIndexReq $, we have that $ > (1-\frac{\ErrorA}{4}) $ proportion of $\BlockElem_{\PrefixElemIndex} \in \RemElemsDefer_{\OuterIterIndex,\InnerIterIndex} \setminus \Prefix_{\PrefixElemIndex-1} $ are \emph{good}. Thus, in terms of the Bernoulli trials $ \DepTrial_{\PrefixElemIndex} $, we have
			\begin{align}
				\Pr[\DepTrial_{\PrefixElemIndex} = 1 \mid \DepTrial_{1} = \depTrial_{1}, \dots, \DepTrial_{\PrefixElemIndex-1} = \depTrial_{\PrefixElemIndex-1}, \PrefixElemIndex \leq \PrefixIndexMax ] > 1 - \frac{\ErrorA}{4}\,. \notag
			\end{align}
			
			Now let $ \IndepTrial_{\PrefixElemIndex} $ be an \emph{independent} Bernoulli Trial, where $ \Pr[\IndepTrial_{\PrefixElemIndex} = 1] = 1 - \frac{\ErrorA}{4} $ and $ \Pr[\IndepTrial_{\PrefixElemIndex} = 0] = \frac{\ErrorA}{4} $. Below, we use \cref{lemma:2nd lemma for bernoulli trials} to relate the dependent and independent Bernoulli trials. From there, we prove \cref{claim:prob ub 1/4} via Markov's inequality.
			\begin{align}
				&\Pr\left[\text{$ < (1-\ErrorA) |\Prefix_{\PrefixIndexMax}| $ elements in $ \Prefix_{\PrefixIndexMax} $ are good}\right] \notag \\
				&= \Pr\left[ \sum_{\PrefixElemIndex=1}^{\PrefixIndexMax} \DepTrial_{\PrefixElemIndex} < (1-\ErrorA)\PrefixIndexMax \right] \notag \\
				&\leq \Pr\left[ \sum_{\PrefixElemIndex=1}^{\PrefixIndexMax} \IndepTrial_{\PrefixElemIndex} < (1-\ErrorA)\PrefixIndexMax \right] &\text{\cref{lemma:2nd lemma for bernoulli trials}} \notag \\
				&= \Pr\left[ \PrefixIndexMax - \sum_{\PrefixElemIndex=1}^{\PrefixIndexMax} \IndepTrial_{\PrefixElemIndex} > \ErrorA\PrefixIndexMax \right] \notag \\
				&= \sum_{\PrefixIndex} \Pr[\PrefixIndexMax = \PrefixIndex] \cdot \Pr\left[ \PrefixIndex - \sum_{\PrefixElemIndex=1}^{\PrefixIndex} \IndepTrial_{\PrefixElemIndex} > \ErrorA\PrefixIndex \mid \PrefixIndexMax = \PrefixIndex \right] &\text{Law of Total Probability} \notag \\
				&\leq \sum_{\PrefixIndex} \Pr[\PrefixIndexMax = \PrefixIndex] \cdot \frac{1}{\ErrorA\PrefixIndex}\mathbb{E}\left[ \PrefixIndex - \sum_{\PrefixElemIndex=1}^{\PrefixIndex} \IndepTrial_{\PrefixElemIndex} \right] &\text{Markov's inequality} \notag \\
				&= \sum_{\PrefixIndex} \Pr[\PrefixIndexMax = \PrefixIndex] \cdot \frac{1}{\ErrorA\PrefixIndex}\left(\PrefixIndex - \left(1-\frac{\ErrorA}{4}\right)\PrefixIndex \right) &\text{$ \mathbb{E}[\IndepTrial_{\PrefixElemIndex}] = 1 - \frac{\ErrorA}{4} $} \notag \\
				&= \sum_{\PrefixIndex} \Pr[\PrefixIndexMax = \PrefixIndex] \cdot \frac{1}{4} = \frac{1}{4}\,. \notag
			\end{align} 
		\end{proof}
		
		By linking \cref{claim:prob ub num good elements,claim:prob ub 1/4}, we have $\Pr[ |\RemElemsDefer_{\OuterIterIndex, \InnerIterIndex}(\Prefix_{\PrefixIndexBest})| > (1-\frac{\ErrorA}{4})|\RemElemsDefer_{\OuterIterIndex, \InnerIterIndex}| ] \leq \frac{1}{4}$, and so $\Pr[\textnormal{\PrefixFailure}] \leq \frac{1}{4}$. 
	\end{proof}
	
	\begin{lemma} \label{lemma:TBS inner failure}
		For every $ \OuterIterIndex $th outer iteration and $ \InnerIterIndex $th inner iteration, $ \Pr[\textnormal{\InnerFailure}] \leq \frac{1}{2} $.
	\end{lemma}
	
	\begin{proof}
		Recall that a necessary condition for an \innerFailure{} is that either a \sampleFailure{} or a \prefixFailure{} occurs. Thus, by a union bound we have
		\begin{align}
			\Pr[\text{\InnerFailure}] & \leq \Pr[\text{\SampleFailure}] + \Pr[\text{\PrefixFailure}] \notag \\
			& \leq \frac{1}{4} + \frac{1}{4} = \frac{1}{2}\,. &\text{\cref{lemma:TBS sample failure,lemma:TBS prefix failure}} \notag
		\end{align} 
	\end{proof}
	
	\begin{lemma} \label{lemma:TBS outer failure}
		For every $ \OuterIterIndex $th outer iteration, $ \Pr[\textnormal{\OuterFailure}] \leq \frac{1}{2} $.
	\end{lemma}
	
	\begin{proof}
		Let $ \DepTrial_\InnerIterIndex $ be a Bernoulli trial \emph{dependent} on the outcomes of previous trials, where
		\begin{itemize}
			\item $ \DepTrial_\InnerIterIndex = 0 $ denotes $ |\RemElemsDefer_{\OuterIterIndex, \InnerIterIndex}(\Prefix_{\PrefixIndexBest})| > (1-\frac{\ErrorA}{4})|\RemElemsDefer_{\OuterIterIndex, \InnerIterIndex}| $ and $ |\RemElemsDefer_{\OuterIterIndex, \InnerIterIndex}| > (1-\ReduProp)|\RemElems_{\OuterIterIndex}| $.
			\item $ \DepTrial_\InnerIterIndex = 1 $ denotes $ |\RemElemsDefer_{\OuterIterIndex, \InnerIterIndex}(\Prefix_{\PrefixIndexBest})| \leq (1-\frac{\ErrorA}{4})|\RemElemsDefer_{\OuterIterIndex, \InnerIterIndex}| $ or $ |\RemElemsDefer_{\OuterIterIndex, \InnerIterIndex}| \leq (1-\ReduProp)|\RemElems_{\OuterIterIndex}| $.
		\end{itemize}
		
		That is, $ \DepTrial_\InnerIterIndex $ indicates whether the $ \InnerIterIndex $th inner iteration successfully caused $ \geq \frac{\ErrorA}{4} $ proportion of elements in $ \RemElemsDefer_{\OuterIterIndex, \InnerIterIndex} $ to be implicitly filtered out from $ \RemElemsDefer_{\OuterIterIndex, \InnerIterIndex+1} $. Observe that a necessary condition for $ \DepTrial_\InnerIterIndex = 0 $ is an \innerFailure{} since $ \DepTrial_\InnerIterIndex = 0 $ implies either a \prefixFailure{} when $ |\BlockFiltered| \geq \frac{\ErrorA}{4}|\RemElems_{\OuterIterIndex}| $, or a \sampleFailure{} when $ |\BlockFiltered| < \frac{\ErrorA}{4}|\RemElems_{\OuterIterIndex}| $.
		\begin{align}
			&\Pr[\DepTrial_{\InnerIterIndex} = 1 \mid \DepTrial_{1} = \depTrial_{1}, \dots, \DepTrial_{\InnerIterIndex-1} = \depTrial_{\InnerIterIndex-1}] \notag \\
			& \geq 1 - \Pr[\text{\InnerFailure}] \geq \frac{1}{2}\,. &\text{\cref{lemma:TBS inner failure}} \notag
		\end{align}
		
		Now recall that, by definition, an \outerFailure{} occurs in the $ \OuterIterIndex $th outer iteration when $ |\RemElems_{\OuterIterIndex+1}| >  (1-\ReduProp)|\RemElems_{\OuterIterIndex}| $. So a necessary condition for an \outerFailure{} to occur is that an insufficient number of inner iterations $ \InnerIterIndex $ have $ |\RemElemsDefer_{\OuterIterIndex, \InnerIterIndex}(\Prefix_{\PrefixIndexBest})| \leq (1-\frac{\ErrorA}{4})|\RemElemsDefer_{\OuterIterIndex, \InnerIterIndex}| $. In terms of the Bernoulli trials $\DepTrial_{\InnerIterIndex}$, the condition is given by $( 1 - \frac{\ErrorA}{4} )^{\sum_{\InnerIterIndex=1}^{\NumBlocks} \DepTrial_\InnerIterIndex} > 1-\ReduProp$, which we rewrite into \cref{eqn:num inner for failure} below.
		\begin{align}
			\sum_{\InnerIterIndex=1}^{\NumBlocks} \DepTrial_\InnerIterIndex < \left\lceil \log_{1-\frac{\ErrorA}{4}}(1-\ReduProp) \right\rceil\,. \label{eqn:num inner for failure}
		\end{align}
		
		Now let $ \IndepTrial_{\InnerIterIndex} $ be an \emph{independent} Bernoulli Trial, where $ \Pr[\IndepTrial_{\InnerIterIndex} = 1] = \frac{1}{2} $ and \mbox{$ \Pr[\IndepTrial_{\InnerIterIndex} = 0] = \frac{1}{2} $}. Below, we use \cref{lemma:1st lemma for bernoulli trials} to relate the dependent and independent Bernoulli trials for inner iterations. From there, we use the Chernoff bound from \cref{lemma:chernoff bounds}, \cref{eqn:chernoff bound less}, assigning $ \Dev = \frac{\ErrorA+2}{\ErrorA+4} $ and $ \mu =  2 ( 1 + \frac{4}{\ErrorA} ) \log(\frac{2}{1-\ReduProp}) \leq \frac{\NumBlocks}{2} $. Thus, we show that $ \Pr[\text{\OuterFailure}] \leq \frac{1}{2} $.
		\begin{align}
			&\Pr[\text{\OuterFailure}] \notag \\
			&\leq \Pr\left[ \sum_{\InnerIterIndex=1}^{\NumBlocks} \DepTrial_\InnerIterIndex < \left\lceil \log_{1-\frac{\ErrorA}{4}}(1-\ReduProp) \right\rceil \right] &\text{\cref{eqn:num inner for failure} condition} \notag \\
			&\leq \Pr\left[ \sum_{\InnerIterIndex=1}^{\NumBlocks} \IndepTrial_\InnerIterIndex < \left\lceil \log_{1-\frac{\ErrorA}{4}}(1-\ReduProp) \right\rceil \right] &\text{\cref{lemma:1st lemma for bernoulli trials}} \notag \\
			&\leq \Pr\left[ \sum_{\InnerIterIndex=1}^{\NumBlocks} \IndepTrial_\InnerIterIndex < \frac{4}{\ErrorA} \log\left(\frac{2}{1-\ReduProp}\right) \right] \notag \\
			&= \Pr\left[ \sum_{\InnerIterIndex=1}^{\NumBlocks} \IndepTrial_\InnerIterIndex < \left(\frac{2}{\ErrorA+4}\right) 2 \left( \frac{\Error + 4}{\ErrorA} \right) \log\left(\frac{2}{1-\ReduProp}\right) \right] \notag \\
			&= \Pr\left[ \sum_{\InnerIterIndex=1}^{\NumBlocks} \IndepTrial_\InnerIterIndex < \left(1-\frac{\ErrorA+2}{\ErrorA+4}\right) 2 \left( 1 + \frac{4}{\ErrorA} \right) \log\left(\frac{2}{1-\ReduProp}\right) \right] \notag \\
			&\leq e^{ - \left( \frac{\ErrorA+2}{\ErrorA+4} \right)^2 \left( 1 + \frac{4}{\ErrorA} \right) \log\left(\frac{2}{1-\ReduProp}\right) } &\text{\cref{lemma:chernoff bounds}} \notag \\
			&= e^{ -\frac{(\ErrorA+2)^2}{\ErrorA(\ErrorA+4)} \log\left(\frac{2}{1-\ReduProp}\right) } \leq e^{ -\log\left(\frac{2}{1-\ReduProp}\right) } = \frac{1-\ReduProp}{2} \leq \frac{1}{2}\,. \notag
		\end{align} 
	\end{proof}
	
	\begin{lemma} \label{lemma:TBS success probability}
		Suppose \TBSa{} (\cref{alg:TBS}) is run such that $\Funcg$ is monotone submodular, $|\GoodElems| \leq \MinSizeForSample$, and $0 < \ProbFail < 1$. Then \TBSa{} successfully terminates with probability $1-\frac{\ProbFail}{\MinSizeForSample}$.
	\end{lemma}
	
	\begin{proof}
		To prove \cref{lemma:TBS success probability}, we show that $ \Pr[\text{\TBSFailure}] \leq \frac{\ProbFail}{\MinSizeForSample} $.
		
		Let $ \DepTrial_\OuterIterIndex $ be a Bernoulli trial \emph{dependent} on the outcomes of previous trials, where
		\begin{itemize}
			\item $ \DepTrial_\OuterIterIndex = 0 $ denotes $ |\RemElemsDefer_{\OuterIterIndex+1}| > (1-\ReduProp)|\RemElemsDefer_{\OuterIterIndex}| $.
			\item $ \DepTrial_\OuterIterIndex = 1 $ denotes $ |\RemElemsDefer_{\OuterIterIndex+1}| \leq (1-\ReduProp)|\RemElemsDefer_{\OuterIterIndex}| $.
		\end{itemize}
		
		That is, $ \DepTrial_\OuterIterIndex $ indicates whether the $ \OuterIterIndex $th outer iteration successfully caused $ \geq \ReduProp $ proportion of elements in $ \RemElemsDefer_{\OuterIterIndex} $ to be filtered out from $ \RemElemsDefer_{\OuterIterIndex+1} $. Observe that an \text{\outerFailure} is, by definition, the event where $ \DepTrial_\OuterIterIndex = 0 $. Thus, we have that
		\begin{align}
			&\Pr[\DepTrial_{\OuterIterIndex} = 1 \mid \DepTrial_{1} = \depTrial_{1}, \dots, \DepTrial_{\OuterIterIndex-1} = \depTrial_{\OuterIterIndex-1}] \notag \\
			& \geq 1 - \Pr[\text{\OuterFailure}] \geq \frac{1}{2}\,. &\text{\cref{lemma:TBS outer failure}} \notag
		\end{align}
		
		Now recall that, by definition, an \TBSFailure{} occurs when $ |\RemElems_{\OuterIterIndex}| > 0 $ by the final check of the \cref{line:TBS check rem elems} if-condition. So a necessary condition for a \TBSFailure{} is that an insufficient number of outer iterations $ \OuterIterIndex $ have $ |\RemElemsDefer_{\OuterIterIndex+1}| \leq (1-\ReduProp)|\RemElemsDefer_{\OuterIterIndex}| $. In terms of the Bernoulli trials $\DepTrial_{\OuterIterIndex}$, the condition is that $|\GoodElems| ( 1 - \ReduProp )^{\sum_{\OuterIterIndex=1}^{\NumIters} \DepTrial_{\OuterIterIndex} } \geq 1$, which we rewrite into \cref{eqn:num outer for TBS failure} below.
		\begin{align}
			\sum_{\OuterIterIndex=1}^{\NumIters} \DepTrial_{\OuterIterIndex} \leq \left\lceil \log_{\frac{1}{1-\ReduProp}} ( |\GoodElems| ) \right\rceil\,. \label{eqn:num outer for TBS failure}
		\end{align}
		
		Now let $ \IndepTrial_{\OuterIterIndex} $ be an \emph{independent} Bernoulli Trial, where $ \Pr[\IndepTrial_{\OuterIterIndex} = 1] = \frac{1}{2} $ and $ \Pr[\IndepTrial_{\OuterIterIndex} = 0] = \frac{1}{2} $. Below, we use \cref{lemma:1st lemma for bernoulli trials} to relate the dependent and independent Bernoulli trials for outer iterations. From there, we use the Chernoff bound from \cref{lemma:chernoff bounds}, \cref{eqn:chernoff bound less}, assigning $ \Dev = \frac{\ReduProp+1/2}{\ReduProp+1} $ and $ \mu = 2 ( 1 + \frac{1}{\ReduProp} ) \log ( \frac{\MinSizeForSample}{\ProbFail} ) \leq \frac{\NumIters}{2} $. Thus, we finally show that $ \Pr[\text{\TBSFailure}] \leq \frac{\ProbFail}{\MinSizeForSample} $.
		\begin{align}
			&\Pr[\text{\TBSFailure}] \notag \\
			&\leq \Pr\left[ \sum_{\OuterIterIndex=1}^{\NumIters} \DepTrial_{\OuterIterIndex} \leq \left\lceil \log_{\frac{1}{1-\ReduProp}} ( |\GoodElems| ) \right\rceil \right] &\text{\cref{eqn:num outer for TBS failure} condition} \notag \\
			&\leq \Pr\left[ \sum_{\OuterIterIndex=1}^{\NumIters} \IndepTrial_{\OuterIterIndex} \leq \left\lceil \log_{\frac{1}{1-\ReduProp}} ( |\GoodElems| ) \right\rceil \right] &\text{\cref{lemma:1st lemma for bernoulli trials}} \notag \\
			&\leq \Pr\left[ \sum_{\OuterIterIndex=1}^{\NumIters} \IndepTrial_{\OuterIterIndex} \leq \left\lceil \log_{\frac{1}{1-\ReduProp}} \left( \frac{\MinSizeForSample}{\ProbFail} \right) \right\rceil \right] &\begin{aligned}
				&&\text{$ |\GoodElems| \leq \MinSizeForSample $ and} \\
				&&\text{$ 0 < \ProbFail < 1 $}
			\end{aligned} \notag \\
			&\leq \Pr\left[ \sum_{\OuterIterIndex=1}^{\NumIters} \IndepTrial_{\OuterIterIndex} \leq \frac{1}{\ReduProp} \log \left( \frac{\MinSizeForSample}{\ProbFail} \right) \right] \notag \\
			&= \Pr\left[ \sum_{\OuterIterIndex=1}^{\NumIters} \IndepTrial_{\OuterIterIndex} \leq \left(\frac{1/2}{\ReduProp+1}\right)2\left(\frac{\ReduProp + 1}{\ReduProp}\right) \log \left( \frac{\MinSizeForSample}{\ProbFail} \right) \right] \notag \\
			&= \Pr\left[ \sum_{\OuterIterIndex=1}^{\NumIters} \IndepTrial_{\OuterIterIndex} \leq \left(1-\frac{\ReduProp+1/2}{\ReduProp+1}\right)2\left(1+\frac{1}{\ReduProp}\right) \log \left( \frac{\MinSizeForSample}{\ProbFail} \right) \right] \notag \\
			&\leq \exp^{ -\left(\frac{\ReduProp+1/2}{\ReduProp+1}\right)^2 \left( 1 + \frac{1}{\ReduProp} \right) \log \left( \frac{\MinSizeForSample}{\ProbFail} \right) } &\text{\cref{lemma:chernoff bounds}} \notag \\ 
			&= \exp^{ -\frac{(\ReduProp+1/2)^2}{\ReduProp(\ReduProp+1)} \log \left( \frac{\MinSizeForSample}{\ProbFail} \right)} \leq \exp^{ -\log\left( \frac{\MinSizeForSample}{\ProbFail} \right) } = \frac{\ProbFail}{\MinSizeForSample}\,. \notag
		\end{align} 
	\end{proof}
	
	\paragraph{Marginal Gain Properties of \TBSa{}.}
	\begin{lemma} \label{lemma:TBS set value}
		Suppose \TBSa{} (\cref{alg:TBS}) is run such that $\Funcg$ is monotone submodular and $0 < \ErrorA < 1$. Further, suppose that \TBSa{} terminates successfully. Then \TBSa{} returns a set $ \SetT $ satisfying
		\begin{align}
			\Funcg(\SetT \mid \varnothing) \geq \frac{1-\ErrorA}{1+\ErrorA}\Threshold |\SetT|\,. \notag
		\end{align}
	\end{lemma}
	
	\begin{proof}
		Over an entire run of \TBSa{}, let $ \NumSets $ be the total number of sets $ \Prefix_{\PrefixIndexBest} $ that were added into $\SetT$ in \cref{line:TBS T update}, let $ \Prefix_{\PrefixIndexBest, \SetIndex} $ be the $ \SetIndex $th such set added, and let $ \SetT_{\SetIndex} = \Prefix_{\PrefixIndexBest, 1} \cup \dots \cup \Prefix_{\PrefixIndexBest, \SetIndex} $. We also let $\Prefix_{\PrefixIndex, \SetIndex} $ be the prefix, indexed at $\PrefixIndex$, considered in the same iteration as $\Prefix_{\PrefixIndexBest, \SetIndex}$.
		
		Each $ \Prefix_{\PrefixIndexBest, \SetIndex} $ is either the successor of, or the same as, some $ \Prefix_{\PrefixIndex, \SetIndex} $ that satisfies $ \Funcg(\Prefix_{\PrefixIndex, \SetIndex} \mid \SetT_{\SetIndex-1}) \geq (1-\ErrorA)\Threshold |\Prefix_{\PrefixIndex, \SetIndex}| $ (\crefrange{line:TBS good prefixes}{line:TBS assign ibest max}). Either way, we have that $ |\Prefix_{\PrefixIndex, \SetIndex}| \geq \frac{|\Prefix_{\PrefixIndexBest, \SetIndex}|}{ 1+\ErrorA} $ since the prefixes have geometrically increasing size (\cref{line:TBS prefixes}). Thus, we prove \cref{lemma:TBS set value} below.
		\begin{align}
			\Funcg(\SetT \mid \varnothing) &= \sum_{\SetIndex=1}^{\NumSets} \Funcg(\Prefix_{\PrefixIndexBest, \SetIndex} \mid \SetT_{\SetIndex-1}) &\text{telescoping series} \notag \\
			&\geq \sum_{\SetIndex=1}^{\NumSets} \Funcg(\Prefix_{\PrefixIndex, \SetIndex} \mid \SetT_{\SetIndex-1}) &\text{monotonicity} \notag \\
			&\geq \sum_{\SetIndex=1}^{\NumSets} (1-\ErrorA)\Threshold |\Prefix_{\PrefixIndex, \SetIndex}| &\text{\cref{line:TBS good prefixes}} \notag \\
			&\geq \sum_{\SetIndex=1}^{\NumSets} (1-\ErrorA)\Threshold \frac{|\Prefix_{\PrefixIndexBest, \SetIndex}|}{ 1+\ErrorA} &\text{geometric prefix size} \notag \\
			&= \frac{1-\ErrorA}{1+\ErrorA}\Threshold  \sum_{\SetIndex=1}^{\NumSets}|\Prefix_{\PrefixIndexBest, \SetIndex}| \notag \\
			&= \frac{1-\ErrorA}{1+\ErrorA}\Threshold |\SetT|\,. \notag
		\end{align} 
	\end{proof}
	
	\begin{lemma} \label{lemma:TBS rem elements}
		Suppose \TBSa{} (\cref{alg:TBS}) is run with value oracle $\Funcg$, cardinality constraint $\CardA$, and marginal gain threshold $\Threshold$. Further, suppose that \TBSa{} terminates successfully and that it outputs a set $\SetT$ satisfying $|\SetT| < \CardA$. Then for all $ \Elem \in \GoodElems \colon \Funcg(\Elem \mid \SetT) < \Threshold $.
	\end{lemma}
	\begin{proof}
		If \TBSa{} terminates successfully and returns a set $\SetT_{\OuterIterIndex-1, \InnerIterIndex}$ satisfiying $|\SetT_{\OuterIterIndex-1, \InnerIterIndex}| < \CardA$, then the only way to return $\SetT_{\OuterIterIndex-1, \InnerIterIndex}$ is inside the \cref{line:TBS check rem elems} if-block. To enter this if-block, $|\RemElems_{\OuterIterIndex}| = 0$ must hold. Therefore, every $\Elem \in \GoodElems$ must have previously been filtered out of $\RemElems_{\OuterIterIndex}$ in \cref{line:TBS filtering step} and so must satisfy $\Funcg(\Elem \mid \SetT_{\OuterIterIndex-1, \InnerIterIndex}) < \Threshold$. 
	\end{proof}
	
	\section{Simple Parallel Algorithm for $\PSep$-Superseparable \SMCC{}}\label{sec:simple}
	In this section, we describe \LALST{} (\LALSTa{}), the simple algorithm for $\PSep$-superseparable \SMCC{}. We give its pseudocode in \cref{alg:LAT} and state its performance guarantees in \cref{theorem:simple results}.
	
	\begin{theorem} \label{theorem:simple results}
		Let $(\Funcf, \Card)$ be an instance of \SMCC{} where $\Funcf$ is $\PSep$-superseparable. Suppose \LALST{} is run such that $0 < \ApproxParamA < 1$ and $0 < \Error < \frac{1}{2}$. Then, with probability $1 - O(\frac{1-\ApproxParamA}{\PSep})$, \LALST{} achieves:
		\begin{itemize}
			\item a solution $\SetS$ satisfying $|\SetS| \leq \Card$ and $\Funcf(\SetS) \geq \ApproxParamA (5 + \frac{4(2-\Error)\Error}{(1-\Error)(1-2\Error)} )^{-1} \OptValue $,
			\item an adaptive complexity of $O( \frac{1}{\Error^3} \log( \frac{\PSep}{1-\ApproxParamA} ) )$, and
			\item an expected query complexity of $ O( \frac{1}{1-\ApproxParamA} ( \frac{\PSep \Card}{\Error^3} + \frac{\PSep}{\Error^4} ) ) $.
		\end{itemize}
	\end{theorem}
	
	\paragraph{Description of \LALSTa{}.} \LALSTa{} simply constructs the set of $\TopElem$ of top-$ \lceil \frac{\PSep\Card}{1-\ApproxParamA} + \Card \rceil $ valued elements in $\GroundSet$, runs the existing \LALS{} (\LALSa{}) procedure \cite{chen2021best} on $\TopElem$, and returns the solution from this procedure. This approach is interesting due to the $O(\frac{1}{\Error^3} \log( \frac{\GroundSetSize}{\Card} ) )$ adaptive complexity of \LALSa{}; since we run this on $\TopElem$, substituting $\GroundSetSize = |\TopElem| \in O( \frac{\PSep\Card}{1-\ApproxParamA} )$ cancels out the dependence on $\Card$ in the adaptive complexity. We mention that \LALSa{} currently has the lowest adaptive complexity of any constant-factor approximation algorithm for \SMCC{}.
	
	\paragraph{Deriving \cref{theorem:simple results}.} By Theorem 5 of \cite{chen2021best}, with probability $1-\frac{\Card}{\GroundSetSize}$, \LALSa{} achieves: 1) a solution $\SetS$ satisfying $|\SetS| \leq \Card$ and $\Funcf(\SetS) \geq (5 + \frac{4(2-\Error)\Error}{(1-\Error)(1-2\Error)} )^{-1} \OptValue $, 2) an adaptive complexity of $O(\frac{1}{\Error^3} \log( \frac{\GroundSetSize}{\Card} ) )$, and 3) an expected query complexity of  $O( (\frac{1}{\Error \Card} + 1 ) \frac{\GroundSetSize}{\Error^3} )$.
	
	Since \LALSTa{} runs \LALSa{} on the top-valued elements $\TopElem$, the approximation factor is only worsened by a factor of $\ApproxParamA$ by \cref{lemma:top elems opt value} and the $\PSep$-superseparability of $\Funcf$. The remaining performance guarantees follow by substituting $\GroundSetSize \in O( \frac{\PSep\Card}{1-\ApproxParamA} )$ into the guarantees of \LALSa{}.
	
	\begin{algorithm}
		\caption{} \label{alg:LAT}
		\begin{algorithmic}[1]
			\Procedure{\LALST}{$\Funcf, \GroundSet, \PSep, \Card, \ApproxParamA, \Error$}
			\Input{value oracle $ \Funcf \colon 2^\GroundSet \rightarrow \mathbb{R}_{\geq 0} $, ground set $ \GroundSet $, parameter $ \PSep $ such that $ \Funcf $ is $ \PSep $-superseparable, cardinality constraint $ \Card $, approximation term $ \ApproxParamA $, approximation error $\Error$}
			\Output{set $ \SetS $ satisfying $\Funcf(\SetS) \geq \ApproxParamA \left(5 + \frac{4(2-\Error)\Error}{(1-\Error)(1-2\Error)} \right)^{-1} \OptValue $}
			\State $ \TopElem \gets $ set of top-$ \left\lceil \frac{\PSep\Card}{1-\ApproxParamA} + \Card \right\rceil $ elements $ \Elem \in \GroundSet $ by value $ \Funcf(\Elem) $ \label{line:LAT assign TopElem}
			\State $ \SetS \gets \LALS(\Funcf, \TopElem, \Card, \Error) $ \label{line:LAT run LinearSeq}
			\State \Return{$ \SetS $} \label{line:LAT return solution}
			\EndProcedure
		\end{algorithmic}
	\end{algorithm}
	
	\section{Conclusions} \label{sec:conclusions}
	In this paper, we propose highly parallel algorithms for $\PSep$-superseparable \SMCC{} that achieve adaptive complexities independent of $\GroundSetSize$, but dependent on parameters $\PSep$ and $\Card$, with the main algorithm being \LSTPGSa{}. We also propose a new procedure \TBS{}, a subroutine of \PGS{}, which is the key to improving the existing state-of-the-art query complexity of our \LSTPGSa{}, not only for the $\PSep$-superseparable \SMCC{}, but also for the general case.
	An interesting research direction is to 
	design an algorithm whose adaptivity depends on $\PSep$ and $\Card$ without the need of prior knowledge on the value of $\PSep$, as our \LSTPGSa{} needs to know this value to set parameters appropriately.
	Also, our simple algorithm for $\PSep$-superseparable \SMCC{} hints at the possibility of a $(5 + O(\Error))^{-1}$-approximation algorithm that only requires $O(\frac{1}{\Error^3} \log \PSep )$ rounds, removing the $\ApproxParamA$ term in the current approximation factor. Finally, it is also worth conducting experiments to compare our algorithms against the existing parallel algorithms for general \SMCC{}, especially on those submodular $\PSep$-superseparable functions with small values of $\PSep$.
	
	\paragraph{Acknowledgements.} This work was in part supported by ARC Discovery Early Career Researcher Award (DECRA) DE190101118 and the University of Melbourne Faculty of Engineering and Information Technology, and School of Computing and Information Systems.
	
	\bibliographystyle{splncs04}
	\bibliography{Fast_Parallel_Algorithms_for_Submodular_p-Superseparable_Maximization}
	
\end{document}